\documentclass[a4paper,12pt]{article}
\usepackage{amsfonts}
\usepackage{amsmath}
\usepackage{amssymb}
\usepackage{bbm}
\usepackage{kotex}
\usepackage{graphicx}
\usepackage{ulem}
\usepackage{cancel}
\usepackage{color}

\newtheorem{theorem}{Theorem}

\newtheorem{corollary}[theorem]{Corollary}

\newtheorem{definition}[theorem]{Definition}

\newtheorem{lemma}[theorem]{Lemma}

\newtheorem{proposition}[theorem]{Proposition}
\newtheorem{remark}[theorem]{Remark}

\newenvironment{proof}[1][Proof]{\textbf{#1.} }{\ \rule{0.5em}{0.5em}}
\scrollmode

\newcommand{\unit}{\hbox{\rm 1\kern-2.8truept l}}

\newcommand{\tr}{\hbox{\rm tr}}
\newcommand{\md}{\hbox{\rm d}}
\newcommand{\me}{\hbox{\rm e}}
\newcommand{\mi}{\mathrm{i}}

\newcommand{\XY}{X\kern-9trueptY}

\begin{document}
\title{Environment induced entanglement in Gaussian open quantum systems}
\author{A. Dhahri$^{(1)}$, F. Fagnola$^{(1)}$, D. Poletti$^{(2)}$, H.J. Yoo$^{(3)}$}

\maketitle

\begin{abstract}
We show that a bipartite Gaussian quantum system interacting with an external Gaussian environment may possess
a unique Gaussian entangled stationary state and that any initial state converges towards this stationary state.
We discuss dependence of entanglement on temperature and interaction strength and show that one can find
entangled stationary states only for low temperatures and weak interactions.
\end{abstract}

{\it Alternative title: Entangled stationary states in Gaussian open quantum systems}


\section{Introduction}

The theory of open quantum systems is a fundamental tool of research in quantum mechanics and its applications because,
in reality, every system interacts with an external environment. The effect of the interaction is generally that of a
disturbing noise that may be small, but it can hardly be neglected. Luckily, recent results suggest that one can also
take advantage of the presence of noise in some efficient way to generate or enhance certain properties of the system states
such as entanglement (see \cite{BCFS,BFP,BeFl,HaackJoye,TicVio} and the references therein).

Entanglement is a distinctive feature of quantum states and a key resource in quantum information,
computing and communication. Generating and stabilizing entangled states is a difficult task.
A strategy known as quantum reservoir engineering where the coupling to the environment provides the
dissipation required for stabilization has been recently investigated (\cite{RoRouSel}).
It has been demonstrated that entangled steady states arise in the presence of dissipation
for finite dimensional (mostly fermion) systems (\cite{BFP,BeFl,HaackJoye,TicVio}).

For infinite-dimensional boson systems, there is doubt whether it is possible to find entangled Gaussian
states that come into play in various important quantum information processing activities such as entanglement distillation.
In \cite{LJB} it has been shown that entangled states cannot be obtained when a system with strictly positive quadratic hamiltonian
is weakly coupled with a reservoir. In this paper we show that, dropping this condition,
entanglement generation and stabilization is possible by appropriate modulation of squeezing parameters.

A big challenge when investigating open quantum systems is to characterize invariant states
when they are coupled to environments that drive the system out of equilibrium (see
\cite{AcFaQu,BoRCUr,JPW}). However, in the case of Gaussian Markov systems,
explicit formulas (\cite{AFPOSID22,DeVaVe,GJN,Po2022,Teret}) allow one to write invariant states and analyze them.
Furthermore, necessary and sufficient conditions (\cite{ViWe,WeWo}) are available to establish whether
a certain state of a bipartite system is entangled or not.

In this article we consider a bipartite Gaussian system consisting of two one-mode subsystems with Hamiltonian
\begin{equation}\label{eq:HS}
H_S = \frac{\omega}{2}\left(a_1 a^\dagger_{2}+a^\dagger_1 a_{2}\right)
+\frac{\kappa}{2}\sum_{j=1}^2\left( a_j^2 + a_j^{\dagger 2}\right)
\end{equation}
where $a_j$ and $a_j^\dagger$ are the usual creation and annihilation operators and $\omega,\kappa$ real constants.
The parameter $\kappa$ is needed to induce squeezing that is a necessary condition for entanglement
of quantum Gaussian states (see \cite{WoEiPl}).

First, we analyze the situation in which only one of the two parties is disturbed by noise.
A mode of the system interacts with a ``reservoir'' consisting of a damped and pumped quantum harmonic oscillator
(see \cite{AlLe} Section 2.2 and equation \eqref{eq:1noiseL1L2} here) with inverse temperature $\beta$.
Applying known results on the asymptotic behaviour of Gaussian open quantum systems (\cite{FP-IDAQP24}),
we show that, if the constant $\kappa\not=0$, the reservoir temperature and the system-reservoir coupling constant
$g\not=0$ are small, the whole three-mode Gaussian system has a unique Gaussian invariant
state and any initial state converges towards this invariant state in trace norm. Then we take its partial trace
onto the bipartite system and show that the bipartite Gaussian state we get is entangled if
$\kappa\not=0$ is not too small, comparable with $\omega$, say, for small reservoir temperatures and small coupling
strength. Theorem \ref{th:bgk} and Corollary \ref{cor:1noise-ent} provide more precise quantitative information.

Second, we analyze the situation in which both parties are disturbed by noise. We show that for $\kappa=0$ one
always gets separable states (Proposition \ref{prop:2noises-k=0})  but, for $\kappa=1$ we get entangled states
for small coupling strength and small temperatures of both reservoirs (Proposition \ref{prop:2noises-k=1}).
In this case, the difference in temperatures does not seem as important as the other parameters.
A more detailed analysis of the case in which the temperatures are the same shows that a new intriguing
feature emerges. Theorem \ref{th:2noise-ent} shows that the Gaussian steady state induced on the bipartite system
is entangled in the following cases: at low temperature for any system-reservoir coupling constant $g\not=0$, at medium/low temperature
for small values of $g$, at medium/high temperature for small but not too small $g$.

We emphasize that, in both situations, external noises drive the system to an entangled stationary state even if its initial
state is separable because any initial state of the system and reservoirs converges towards a unique
invariant state by Theorem \ref{thm:esZCesZ-conv}. Moreover, our analysis of entanglement in terms of interaction
strength and reservoir temperatures is exahaustive; for all their values we determine whether the unique stationary state
is entangled or separable.

The paper is organized as follows. We introduce definitions and notation for Gaussian quantum Markov semigroups
(QMS) in Section \ref{sect:GQMS}. In Section \ref{sect:1noise} we analyze
the model with noise interacting only with one party. The case of two independent noises acting
on both parties is studied in Section \ref{sect:two-noises-model}.

\section{Gaussian states and quantum Markov semigroups}\label{sect:GQMS}

Let $\mathsf{h}=\Gamma(\mathbb{C}^d)$ be the $d$-mode Fock space on $\mathbb{C}^d$  and let
$(e_k)_{1\leq k\leq d}$ be the canonical orthonormal basis of $\mathbb{C}^d$.
The Hilbert space $\mathsf{h}$ is isometrically isomorphic to $\Gamma(\mathbb{C})\otimes\cdots\otimes\Gamma(\mathbb{C})$
and the set of exponential vectors $e_f$, also called coherent vectors, is total in $\mathsf{h}$.
Let $a_j, a_j^\dagger$ be the creation and annihilation operator of the Fock representation of the
$d$-dimensional Canonical Commutation Relations (CCR) defined on the set of exponential vectors by
\[
a_j e_f = \langle e_j, f\rangle e_f, \qquad
a_j^\dagger\,e_f = \frac{\mathrm{d}}{\mathrm{d}r} e_{f+re_j}\Big|_{r=0}.
\]
The above operators are obviously defined on the linear manifold spanned by exponential vectors
that turns out to be an essential domain for all the operators considered in this paper.
This is also the case for field operators
\begin{equation}\label{eq:field-op}
q_j = \left(a_j+a^\dagger_j\right)/\sqrt{2}\qquad p_j = \mi\left(a^\dagger_j-a_j\right)/\sqrt{2}.
\end{equation}

Unitary Weyl operators, are defined on the exponential vectors via the formula
\[
W(z)e_g=\hbox{\rm e}^{-\parallel z\parallel^2/2-\langle z,g \rangle}e_{z+g}\quad z,g\in\mathbb{C}^d.
\]
By this definition $\langle W(z)e_f,W(z)e_g\rangle = \langle e_f,e_g\rangle$ for all $f,g\in\mathbb{C}^d$,
therefore $W(z)$ extends uniquely to a unitary operator on $\mathsf{h}$. Weyl operators satisfy the CCR in the
exponential form, namely, for every $z,w \in \mathbb{C}^d$,
\begin{equation} \label{eq:WeylCCR}
	W(z)W(w) = \me^{-\mi \Im \langle{z},{w}\rangle} W(z+w).
\end{equation}

\begin{definition}
A density matrix $\rho$ is called a quantum Gaussian state if there exist
$\mu \in \mathbb{C}^d$ and a \emph{real} linear, symmetric, invertible operator $S$ such that
	\begin{equation} \label{funz-caratt-rho}
		\hat{\rho}(z) = \tr\left(\rho W(z)\right)= \exp\left( - \frac{1}{2} \Re\left\langle z,{Sz}\right\rangle
              - \mi \Re\left\langle \mu,z\right\rangle  \right),
            \quad \forall z \in \mathbb{C}^d.
	\end{equation}
	In that case $\mu$ is said to be the  mean vector  and $S$ the covariance operator  and we will denote it also with $\rho_{(\mu,S)}$.
\end{definition}

Here $\Re$ and $\Im$ denote the real and imaginary parts of complex numbers, vectors and matrices.
It is useful to identify the real linear operator $S$ on $\mathbb{C}^d$ with a matrix $\mathbf{S}$ with real entries acting on
$\mathbb{R}^{2d}$ and complex vectors $z=x+\mi y\in\mathbb{C}^d$  with $x,y\in\mathbb{R}^d$ with a vector $\mathbf{z}$ in
$\mathbb{R}^{2d}$. More precisely, $\mathbf{z}=[x,y]^{\rm\scriptstyle T}$
(${\rm\scriptstyle T}$ denotes the transpose) and, since $S z = S_1 z+ S_2 \overline{z}$ with $S_1,S_2$ complex linear,
for every $z \in \mathbb{C}$, we write	
\[
		\mathbf{S} \mathbf{z}=
		\left[\begin{array}{cc}
			\Re{S_1} + \Re{S_2} & \Im{S_2}-\Im{S_1} \\
			\Im{S_1}+ \Im{S_2} & \Re{S_1} - \Re{S_2}
		\end{array}\right]
		\left[\begin{array}{c}
		x \\ y
		\end{array}\right]
	\]
Conversely, given $2d\times 2d$ matrix with real entries
\[
	\mathbf{S}= \left[\begin{array}{cc}
		S_{11} & S_{12} \\
		S_{21} & S_{22}
	\end{array}\right],
\]
we can induce a real linear operator $S$ on $\mathbb{C}^d$ via
\[
	S z = \left( (S_{11}+S_{22})/{2} -  (S_{21}-S_{12})/({2\mi})\right) z
+ \left( ({S_{11}-S_{22}})/{2} - (S_{12}+S_{21})/({2\mi})\right) \overline{z}.
\]
It is also useful to introduce a real linear operator $J$ corresponding to the multiplication by $-\mi$
and note that
\[
	Jz= -\mi z, \quad \mathbf{J}=\left[\begin{array}{cc} 0 & \unit_d \\ -\unit_d & 0 \end{array}\right].
\]
where $\unit_d$ is the $d\times d$ identity matrix.

\begin{remark}{\rm
	Recall that a real linear, symmetric, invertible operator $S$ is a suitable covariance operator
of a Gaussian state if and only if
	\begin{equation} \label{eqn:HeisenbergCond}
		\mathbf{S} - \mi \mathbf{J} \geq 0,
	\end{equation}
where $\mathbf{S}$ and $\mathbf{J}$ are now considered as complex linear operator on $\mathbb{C}^{2d}$ (see \cite{KRP-cosa}
Theorem 3.1, with a little warning: $J$ there is the multiplication by $\mi$ instead of $-\mi$).
Therefore positivity is evaluated with respect to the usual complex inner product.}
\end{remark}

With the above notations, the covariance matrix of a zero-mean quantum Gaussian state is related to second order moments
by the following identities
\[
2\operatorname{tr}(\rho\, p_j\, p_k) = \mathbf{S}_{jk},  \quad
 2\operatorname{tr}\operatorname(\rho\,  q_j\, q_k) = \mathbf{S}_{j+d\,k+d}
\qquad
\operatorname{tr}\operatorname(\rho\, \{p_j , q_k\})  = -\mathbf{S}_{j\,k+d}
\]
where $j,k=1,\dots, d $ and $\{\cdot,\cdot\}$ denotes the anti-commutator.

A simple necessary and sufficient condition for entanglement of bipartite Gaussian state is known \cite{WeWo} Proposition 2
that translates at the level of density matrices the partial transpose condition.

If  $R_\alpha$, $\alpha=1,\cdots,2d$, are the canonical operators (position and momentum), the commutation relations are
\[
[R_\alpha,R_\beta]=\mi J_{\alpha \beta}\unit_{2d},
\]
with the symplectic matrix ${\mathbf J}=(J_{\alpha\beta})=\left[\begin{matrix}0&\unit_d\\-\unit_d&0\end{matrix}\right]$.
Suppose now the system consists of two parties $A$ and $B$ of sizes $d_A$ and $d_B$, Alice and Bob  respectively to say, with $d_A+d_B=d$.
A product state of two Gaussian states on $A$ and $B$ with covariance matrics ${\mathbf S}_A$ and ${\mathbf S}_B$ is again a Gaussian state with covariance the direct sum ${\mathbf S}={\mathbf S}_A\oplus {\mathbf S}_B$. Moreover, it was shown in \cite{WeWo} that a Gaussian state with covariance ${\mathbf S}$ of a bipartite system consisting of parties $A$ and $B$ is separable, i.e., unentangled between $A$ and $B$ if and only if there are covariance matrices ${\mathbf S}_A$ and ${\mathbf S}_B$ such that ${\mathbf S}\ge {\mathbf S}_A\oplus {\mathbf S}_B$.

As was shown in \cite[Proposition 2]{WeWo}, a Gaussian state with a covariance $\mathbf S$ has positive partial trace if
${\mathbf S}+\mi\tilde {\mathbf J}\ge 0$,
where ${\tilde{\mathbf J}}=\left[\begin{matrix} -{\mathbf J}_A&0\\0&{\mathbf J}_B\end{matrix}\right]$.
In other words, once this fails we realize an entangled Gaussian state.
Even more explicitly, in the case where $d=2$ of our interest, we have the following.

\begin{theorem}\label{th:bipartite-gauss-ent}
Assume  $d=2$. A bipartite Gaussian state with covariance matrix $\mathbf{S}$ is separable if and only if the matrix
\begin{equation}\label{eq:S-tilde}
\widetilde{\mathbf{S}} = \mathbf{S} +
\left[\begin{array}{cccc}  0 & 0 & \mathrm{i} & 0 \\
                           0 & 0 &  0 & -\mathrm{i} \\
                           -\mathrm{i} & 0 & 0 & 0 \\
                           0 & \mathrm{i} & 0 & 0 \end{array}  \right]
\end{equation}
is positive semidefinite.
\end{theorem}

We refer to \cite{ViWe,WeWo,WoEiPl} for results on entanglement of quantum Gaussian states and to
\cite{Raja-Tiju,CrDVMoRo} for further results, also in the infinite mode case.

A QMS $\mathcal{T}=(\mathcal{T}_t)_{t\geq 0}$ is a weakly$^*$-continuous semigroup of completely positive,
identity preserving, weakly$^*$-continuous maps on $\mathcal{B}(\mathsf{h})$.
The predual semigroup $\mathcal{T}_*= (\mathcal{T}_{*t})_{t\geq 0}$ on the predual space of trace class
operators on $\mathsf{h}$ is a strongly continuous contraction semigroup.

Gaussian QMSs (sometimes also called quasi-free semigroups \cite{DeVaVe}) can be introduced through their generator
which is unbounded but sufficiently well-behaved.
It can be represented in a generalized (since operators $L_\ell, H$ are unbounded) Gorini–Kossakowski–Lindblad-Sudarshan (GKLS)
form (see \cite{Po2022} Theorems 5.1, 5.2) as
\begin{equation}\label{eq:GKLS}
\mathcal{L}(x) = \mi\left[ H, x\right]
-\frac{1}{2}\sum_{\ell=1}^m \left( L_\ell ^*L_\ell\, x - 2 L_\ell^* x L_\ell + x\, L_\ell ^*L_\ell\right).
\end{equation}
where $1 \leq m \leq 2d$, and
	\begin{align}
		H&= \sum_{j,k=1}^d \left( \Omega_{jk} a_j^\dagger a_k + \frac{\kappa_{jk}}{2} a_j^\dagger a_k^\dagger + \frac{\overline{\kappa_{jk}}}{2} a_ja_k \right) + \sum_{j=1}^d \left( \frac{\zeta_j}{2}a_j^\dagger + \frac{\bar{\zeta_j}}{2} a_j \right), \label{eq:H}\\
		 L_\ell &= \sum_{k=1}^d \left( \overline{v_{\ell k}} a_k + u_{\ell k}a_k^\dagger\right),\label{eq:Lell}
	\end{align}
$\Omega:=(\Omega_{jk})_{1\leq j,k\leq d} = \Omega^*$ and $K:= (\kappa_{jk})_{1\leq j,k\leq d}= \kappa^{\hbox{\tiny T}}
\in M_d(\mathbb{C})$, are $d\times d$ complex matrices with $\Omega$ Hermitian and $K$ symmetric,
$V=(v_{\ell k})_{1\leq \ell\leq m, 1\leq  k\leq d}, U=(u_{\ell k})_{1\leq \ell\leq m, 1\leq  k\leq d} \in M_{m\times d}(\mathbb{C})$
are $m\times d$ matrices and $\zeta=(\zeta_j)_{1\leq j\leq d} \in \mathbb{C}^d$.

Note that operators $L_\ell$ are closable therefore we will identify them with their closure.

Clearly, $\mathcal{L}$ is well defined on the dense  sub-$^*$-algebra of $\mathcal{B}(\mathsf{h})$ generated by rank one
operators $|\xi\rangle\langle \xi'|$ with $\xi,\xi'\in D$ but the operators $H,L_\ell$ are unbounded and
the domain of $\mathcal{L}$ is not the whole of $\mathcal{B}(\mathsf{h})$. For this reason we associate
with each $x\in\mathcal{B}(\mathsf{h})$ a quadratic form on $\mathsf{h}$ with domain $D\times D$
\begin{eqnarray}\label{eq:Lform}
\texttt{\textrm{\pounds}}(x)  [\xi',\xi] &= & \mi\langle{H\xi'},{x\xi}\rangle - \mi\langle{\xi'},{x H\xi}\rangle \\
	&- & \frac{1}{2} \sum_{\ell=1}^m \left( \langle{\xi'},{xL_\ell^*L_\ell \xi}\rangle -2\langle{L_\ell \xi'},{xL_\ell \xi}\rangle
    + \langle{L_\ell^* L_\ell \xi'},{x \xi}\rangle \right) \nonumber
\end{eqnarray}
An application of the minimal semigroup method yields the following
\begin{theorem}\label{th:G-QMS!}
There exists a unique QMS, $\mathcal{T}=(\mathcal{T}_t)_{t\geq 0}$ such that, for all $x\in\mathcal{B}(\mathsf{h})$ and
$\xi,\xi'\in D$, the function $t\mapsto \langle{\xi'},{\mathcal{T}_t (x) \xi} \rangle $ is differentiable and
\begin{align*}
	\frac{\hbox{\rm d}}{\hbox{\rm d}t} \langle{\xi'},{\mathcal{T}_t (x) \xi}\rangle
= \texttt{\textrm{\pounds}}(\mathcal{T}_t(x)) [\xi',\xi]
 \qquad \forall\, t\geq 0.
\end{align*}
The domain of the generator consists of $x\in\mathcal{B}(\mathsf{h})$ for which the quadratic form
$\texttt{\textrm{\pounds}}(x)$ is represented by a bounded operator.
\end{theorem}

Weyl operators, in general, do not belong to the domain of the generator of $\mathcal{T}$ because a straightforward
computation shows that the quadratic form $\texttt{\textrm{\pounds}}(W(z))$ is typically unbounded.
However, one has the following explicit formula (see \cite{DeVaVe})
\begin{theorem}\label{th:explWeyl}
Let $(\mathcal{T}_{t})_{t\ge 0}$ be the quantum Markov semigroup with generalized GKLS  generator
associated with $H,L_\ell$ as above. For all Weyl operator $W(z)$ we have
\begin{equation}\label{eq:explWeyl}
\mathcal{T}_t(W(z))
= \exp\left(-\frac{1}{2}\int_0^t \Re\langle{\hbox{\rm e}^{sZ}z},{
  C \hbox{\rm e}^{sZ}z}\rangle\hbox{\rm d}s
+\mi\int_0^t  \Re\langle{\zeta},{\hbox{\rm e}^{sZ}z}\rangle \hbox{\rm d}s \right)
W\left(\hbox{\rm e}^{tZ}z\right)
\end{equation}
where the \emph{real linear} operators $Z,C$ on $\mathbb{C}^d$ are
\begin{eqnarray}
  Zz &=& \left[ \left(   U^{\rm\scriptstyle T}\overline{U} - V^{\rm\scriptstyle T}\overline{V}\right)/2
  + \mi \Omega \right] z + \left[\left( U^{\rm\scriptstyle T}V - V^{\rm\scriptstyle T} U\right)/2
  + \mi K \right]\overline{z}\\
   Cz &=& \left(  U^{\rm\scriptstyle T}\overline{U} + V^{\rm\scriptstyle T}\overline{V}\right) z
   + \left( U^{\rm\scriptstyle T}V + V^{\rm\scriptstyle T} U\right)\overline{z} \label{eq:bfC}
   \end{eqnarray}
\end{theorem}


Let $Z^\sharp$ denote the adjoint of the real linear operator $Z$ with respect to the scalar product
$\Re\langle\cdot,\cdot\rangle$, namely $Z^\sharp z = Z_1^* z + Z_2^T \overline{z}$.
As an immediate consequence of Theorem \ref{th:explWeyl} one can show that the predual QMS of a Gaussian QMS
preserves Gaussian states
\begin{proposition} \label{prop:gaussianStateEvolution}
	If $\mathcal{T}$ is a Gaussian QMS, then $\mathcal{T}_{*t} (\rho_{(\mu, S)}) = \rho_{(\mu_t, S_t)}$ with
	\begin{equation} \label{eq:invariantParameterEvolution}
		\mu_t = \me^{tZ^\sharp}\mu - \int_0^t \me^{sZ^\sharp} \zeta \md s, \quad S_t
= \me^{tZ^\sharp} S \me^{tZ} + \int_0^t \me^{sZ^\sharp} C \me^{sZ} \md s.
	\end{equation}
\end{proposition}
It can be shown (see \cite{Po2022} Theorems 5.1, 5.2) that a QMS  $\mathcal{T}$ is Gaussian if
maps $\mathcal{T}_{*t}$ of the predual semigroup $\mathcal{T}_*$ preserve Gaussian states.

For the purposes of this paper we need only consider zero-mean Gaussian states and Hamiltonian $H$ \eqref{eq:H}
with zero linear part, therefore, from now on, we assume $\mu=\zeta=0$. In this case, if the matrix $\mathbf{Z}$
is stable, namely its spectrum is contained in the left-half plane $\{\, \lambda\in\mathbb{C}\,\mid\, \Re\lambda < 0\,\}$,
it is clear from \eqref{eq:invariantParameterEvolution} that, the QMS $\mathcal{T}$ admits a Gaussian invariant
state with covariance matrix $\int_0^\infty \me^{sZ^\sharp} C \me^{sZ} \md s$.

Somewhat complex but elementary calculations show that the complex linear operators on $\mathbb{C}^{2d}$ determined
by $Z$ and $C$ are
\begin{eqnarray}
\mathbf{Z} & = & \frac{1}{2}\left[\begin{array}{cc}
			\Re \left(\left( U - \overline{V}\right)^*\left(U +\overline{V}\right)\right)
            & \Im\left(\left( U - \overline{V}\right)^*\left(U -\overline{V}\right)\right) \\
			- \Im \left(\left( U + \overline{V}\right)^*\left(U +\overline{V}\right)\right)
            & \Re \left(\left( U + \overline{V}\right)^*\left(U -\overline{V}\right)\right)		\\
           \end{array}\right] \nonumber \\
 &  & + \left[\begin{array}{cc}
			-\Im \left(\Omega + K\right) & \Re\left(K - \Omega \right) \\
			\Re\left(\Omega + K\right) & \Im\left(K - \Omega\right)
		\end{array}\right] \label{eq:mathbfZ} \\
	\mathbf{C} & = & \left[\begin{array}{cc}
			\Re \left( \left(U + \overline{V}\right)^*\left(U + \overline{V}\right)\right)
         & \Im \left( \left(U + \overline{V}\right)^*\left(U - \overline{V}\right)\right)\\
			-\Im \left( \left(U - \overline{V}\right)^*\left(U + \overline{V}\right)\right)
         & \Re \left( \left(U - \overline{V}\right)^*\left(U - \overline{V}\right)\right)
		\end{array}\right] \label{eq:mathbfC}
\end{eqnarray}

It is known that Gaussian QMSs with a stable matrix $\mathbf{Z}$ are, in turn, stable in the sense that each
initial state converges towards a unique invariant state. More precisely, by \cite{FP-IDAQP24} Theorem 9, we have
\begin{theorem}\label{thm:esZCesZ-conv}
If $\mathbf{Z}$ is stable, the QMS $\mathcal{T}$ has a unique zero-mean Gaussian invariant state with covariance matrix
\begin{equation}\label{eq:cov-mean}
\mathbf{S}=\int_0^\infty \me^{s\mathbf{Z}^*}\mathbf{C}\me^{s\mathbf{Z}} \md s
\end{equation}
Moreover, any initial state converges in trace norm towards this Gaussian invariant state. In particular, the Gaussian invariant state is also the unique normal invariant state.
\end{theorem}

The same conclusion can also be obtained from arguments based on irreducibility in \cite{AFPOSID22,FP-IDAQP22}
taking care of domain conditions on $G$.

In concrete models, the simplest way to calculate the covariance matrix $\mathbf{S}$ is by solving
the equation
\begin{equation}\label{eq:ZTS+SZ+C=0}
\mathbf{Z}^{\rm\scriptsize T} \mathbf{S} + \mathbf{S} \mathbf{Z} + \mathbf{C}=0
\end{equation}
In Appendix A we describe how to solve this equation in the cases of our interest.

\section{A single noise model}\label{sect:1noise}

In this section we consider a bipartite system in which only one of the parties interacts with an external noise.
More precisely, the system space is a two-mode Fock space with creation and annihilation operators $a^\dagger_1,a^\dagger_2$
and $a_1,a_2$. Another mode, that we label $0$, plays the role of a noise. Consider the 3-mode Fock space
$\Gamma(\mathbb{C}^3)$ and the Gaussian GKLS generator \eqref{eq:GKLS} with $m=2$ and
\begin{eqnarray}
 & & \qquad L_1 = \left(\frac{\mathrm{e}^{\beta}}{\mathrm{e}^{\beta}-1}\right)^{1/2} a_0
 \qquad
  L_2 = \left(\frac{1}{\mathrm{e}^{\beta}-1}\right)^{1/2} a^\dagger_0  \label{eq:1noiseL1L2}  \\
& & \kern-12truept H = \frac{g}{2}\left(a_0a^\dagger_{1}+a^\dagger_0 a_{1}\right) +
\frac{\omega}{2} \left(a_1 a^\dagger_{2}+a^\dagger_1 a_{2}\right)
+\frac{\kappa}{2}\sum_{j=1}^2\left(a_j^{2}+a^{\dagger 2}_j \right) \label{eq:1noiseH}
 \end{eqnarray}
where $g,\omega$ are non-zero constants so that the interaction between parties 0-1 and 1-2 is non-zero.
The constant $g$ is the coupling strength and $\beta>0$ the inverse temperature of the reservoir.
Note that, with a transformation $a_1\to \mi a_1,\, a^\dagger_1 \to -\mi a^\dagger_1$ and $a_2\to -\mi a_2,\, a^\dagger_2 \to \mi a^\dagger_2$,
if necessary, we can always assume $\omega>0$ and get $\omega=1$ by a time change $t\to \omega t$. Therefore, from now on, we assume $\omega=1$.

In this case, first one immediately identifies matrices  $U,V,\Omega,K$. Then, by \eqref{eq:mathbfZ} and \eqref{eq:mathbfC},
setting $\widetilde{\beta}=\coth(\beta/2)$, matrices $\mathbf{Z}$ and $\mathbf{C}$ of the Gaussian QMS of this model are
\begin{equation}\label{eq:1noise-Z-C}
\mathbf{Z} = \frac{1}{2}\left[ \begin{array}{cccccc}
                      -1 & 0 & 0   & 0 & -g & 0 \\
                      0 & 0 & 0   & -g & \kappa & -1  \\
                      0 & 0 & 0   & 0 & -1 & \kappa  \\
                     0 & g & 0  & -1 & 0 & 0  \\
                     g & \kappa & 1  & 0 & 0 & 0  \\
                     0 & 1 & \kappa  & 0 & 0 & 0  \\
                    \end{array}\right]
\quad
\mathbf{C} =\left[ \begin{array}{cccccc}
                     \widetilde{\beta} & 0 & 0   & 0 & 0 & 0 \\
                      0 & 0 & 0   & 0 & 0 & 0  \\
                      0 & 0 & 0   & 0 & 0 & 0  \\
                     0 & 0 & 0  & \widetilde{\beta} & 0 & 0  \\
                      0 & 0 & 0   & 0 & 0 & 0  \\
                      0 & 0 & 0   & 0 & 0 & 0  \\
                    \end{array}\right]
\end{equation}
Note that $\beta>0$ implies $\widetilde{\beta}>1$. Moreover, the function
$]0,+\infty[\ni t \mapsto \coth(1/(2t))$ monotonically maps $]0,+\infty[$ to $]1,+\infty[$ and
so it is an increasing function of the temperature.
The characteristic polynomial of $\mathbf{Z}$ is the product of two factors
\begin{eqnarray*}
& &
\lambda ^3  +  \frac{\lambda ^2}{2} + \frac{1-\kappa^2+g^2}{4}\lambda + \frac{1-\kappa^2- g^2 \kappa}{8} \\
& & \lambda ^3  +  \frac{\lambda ^2}{2} + \frac{1-\kappa^2+g^2}{4}\lambda + \frac{1-\kappa^2+ g^2 \kappa}{8}
\end{eqnarray*}
By the Routh - Hurwitz criterion (see, for example, \cite{Math24})
a third order polynomial $ c_ 0 \lambda^3 + c_ 1 \lambda^2 + c_ 2 \lambda + c_3$ is stable if and only if
$c_0, c_1, c_3 > 0$ and $c_1 c_2 - c_0 c_3 > 0$. In this case the latter condition becomes
\begin{eqnarray*}
 (1+g^2 -\kappa^2) - (1+g^2 \kappa -\kappa^2) & = &  g^2(1-\kappa) >0 \\
(1+ g^2 -\kappa^2) - (1-g^2 \kappa  -\kappa^2)& = &  g^2(1+\kappa) >0
\end{eqnarray*}
so that, since we assumed $g\not=0$, we find $0\leq |\kappa | < 1$.
Condition $c_3>0$ implies then $ g^2 |\kappa| < 1-\kappa^2$. Summarizing

\begin{lemma}\label{eq:1noise-stab-1}
Assume $g\not=0$. The matrix $\mathbf{Z}$ in \eqref{eq:1noise-Z-C} is stable if and only if
either $\kappa=0$ or
\begin{equation}\label{eq:1noise-stab-Z}
0< |\kappa | < 1 \quad\text{ and }\quad 0 < g^2  < (1-\kappa^2)/|\kappa|.
\end{equation}
\end{lemma}

The above lemma shows that, for $\mathbf{Z}$ to be stable, we need $|\kappa|<1$ and the interaction strength has to
vanish as $|\kappa|$ tends to $1$.

By Theorem \ref{thm:esZCesZ-conv}, solving the equation \eqref{eq:ZTS+SZ+C=0}, we have the following

\begin{proposition}\label{prop:1noise-inv-state}
If $\kappa=0$ or the inequalities \eqref{eq:1noise-stab-Z} hold the Gaussian QMS generated by \eqref{eq:GKLS}
with $L_1,L_2, H$ given by \eqref{eq:1noiseL1L2}, \eqref{eq:1noiseH} has a unique invariant state which is a
quantum Gaussian state with covariance matrix $\mathbf{S}$
equal to $\widetilde{\beta}/((1-\kappa^2)^3-\kappa^2(1-\kappa^2) g^4)$ times
\[
\left[ \begin{array}{cccccc}
{\scriptstyle \delta^3 +g^2\kappa^2(\delta-g^2) } & {\scriptstyle g \kappa^3 (\delta-g^2 ) }
& {\scriptstyle g \kappa^2 (\delta-g^2 )  }
 & {\scriptstyle -g^2 \kappa^3 (\delta-g^2 ) } & {\scriptstyle g k^2 (\delta-g^2\kappa^2 )  }
 & {\scriptstyle -g \kappa (\delta - g^2 \kappa^2) } \\
 {\scriptstyle g \kappa^3 (\delta-g^2 ) } & {\scriptstyle \delta- g^2\kappa^2(g^2+\kappa^2) }
 & {\scriptstyle \kappa (\delta - g^2 \kappa^2) } & {\scriptstyle -g\kappa^2 (\delta - g^2 \kappa^2) }
 & {\scriptstyle \kappa (\delta - g^2 \kappa^2 (g^2+ \kappa^2)) } & {\scriptstyle -\kappa^2 (\delta - g^2 ) } \\
 {\scriptstyle g \kappa^2 (\delta - g^2) } & {\scriptstyle \kappa (\delta - g^2 \kappa^2)}
 &  {\scriptstyle \delta - g^2 \kappa^2} &  {\scriptstyle -g \kappa (\delta - g^2 \kappa^2)}
 &  {\scriptstyle \kappa^2 (1 - \kappa^2-g^2)} &  {\scriptstyle -\kappa (1- \kappa^2-g^2)} \\
  {\scriptstyle -g^2 \kappa^3 (\delta - g^2)} & {\scriptstyle -g\kappa^2 (\delta - g^2 k^2)}
 & {\scriptstyle -g \kappa (\delta - g^2 \kappa^2)} & {\scriptstyle \delta^3 +g^2\kappa^2(\delta-g^2)}
 & {\scriptstyle -g \kappa^3 (\delta - g^2)} & {\scriptstyle g \kappa^2 (\delta - g^2)} \\
 {\scriptstyle g\kappa^2 (\delta - g^2 )} & {\scriptstyle \kappa (\delta - g^2  \kappa^2( \kappa^2 + g^2))}
 & {\scriptstyle \kappa^2 (\delta - g^2)} & {\scriptstyle -g \kappa^3 (\delta-g^2)}
 & {\scriptstyle \delta- g^2\kappa^2(g^2+\kappa^2)} & {\scriptstyle -\kappa (\delta - g^2 \kappa^2)} \\
 {\scriptstyle -g \kappa (\delta - g^2 k^2) } &  {\scriptstyle -\kappa^2 (\delta - g^2) }
 &  {\scriptstyle -\kappa (\delta - g^2) } &  {\scriptstyle g \kappa^2 (\delta- g^2) }
 &  {\scriptstyle -\kappa (\delta - g^2 \kappa^2) } &  {\scriptstyle \delta - g^2 \kappa^2 } \\
\end{array} \right]
\]
where $\delta=1-\kappa^2$.
\end{proposition}

\noindent{\bf Remark.} Note that, for $\kappa=0$, the covariance matrix of the Gaussian invariant state
is $\widetilde{\beta}$ times the identity matrix. Clearly, this is not entangled. Indeed, as expected, the invariant state
is the thermal state $(1-\mathrm{e}^{- \beta})\,\mathrm{e}^{- \beta (a^\dagger_0 a_0+a^\dagger_1 a_1+a^\dagger_2 a_2)}$.

\smallskip

Tracing out the noise, namely removing the first and fourth rows and columns, corresponding to the $0$ mode, the reduced invariant density $\mathbf{S}_{\text{\tiny red}}$  turns out to be $\widetilde{\beta}\left(1-\kappa^2\right)^{-1}$ 
$\cdot\left( (1-\kappa^2)^2- g^4\kappa^2\right)^{-1}$ times
\[
\left[
\begin{array}{cccc}
 {\scriptstyle 1-\kappa^2-g^2\kappa^2(g^2+\kappa^2)}
  &  {\scriptstyle \kappa    \left(1-\kappa^2-g^2\kappa^2 \right)}
  &  {\scriptstyle  \kappa(1-\kappa^2) - g^2\kappa^3\left(1+g^2\right)}
  &   {\scriptstyle -\kappa^2\left(1-\kappa^2-g^2\right)} \\
 {\scriptstyle \kappa    \left(1-\kappa^2-g^2\kappa^2  \right)}
  & {\scriptstyle   1-\kappa^2 -g^2\kappa^2 }
 & {\scriptstyle \kappa^2  \left(1 -\kappa^2 - g^2\right)}
 & {\scriptstyle - k\left(1-\kappa^2-g^2\right)}  \\
  {\scriptstyle \kappa(1-\kappa^2) - g^2\kappa^3\left(1+g^2\right)}
 & {\scriptstyle   \kappa^2\left(1 -\kappa^2 - g^2\right)}
 & {\scriptstyle  1-\kappa^2-g^2\kappa^2(g^2+\kappa^2)}
  & {\scriptstyle -\kappa \left(1-\kappa^2(1+g^2) \right)} \\
 {\scriptstyle  - \kappa^2    \left(1-\kappa^2-g^2\right)}
   & {\scriptstyle -  \kappa\left(1-\kappa^2-g^2\right)}
   & -{\scriptstyle \kappa  \left(1-\kappa^2(1+g^2) \right)}
   & {\scriptstyle   1-\kappa^2- g^2\kappa^2 } \\
\end{array}
\right]
\]
Note that, since $0<|\kappa|<1$ and $g^2<(1-\kappa^2)/|\kappa|$
\begin{eqnarray}
 1-\kappa^2 - g^2\kappa^2 & > & \left( 1-\kappa^2\right)\left( 1-|\kappa|\right) >0 \label{eq:cov-mat-pos-diag-11} \\
 1-\kappa^2-g^2\kappa^2(g^2+\kappa^2)
 & > & 1-\kappa^2-|\kappa|(1-\kappa^2)((1-\kappa^2)|\kappa|^{-1}+\kappa^2) \nonumber \\
 & = & \kappa^2\left( 1-\kappa^2\right)\left( 1-|\kappa|\right) >0 \label{eq:cov-mat-pos-diag-22}
\end{eqnarray}

The reduced state is entangled if and only if the matrix
\[
\widetilde{\mathbf{S}}_{\text{\tiny red}}=
\mathbf{S}_{\text{\tiny red}} +
\left[\begin{array}{cccc}  0 & 0 & \mathrm{i} & 0 \\
                           0 & 0 &  0 & -\mathrm{i} \\
                           -\mathrm{i} & 0 & 0 & 0 \\
                           0 & \mathrm{i} & 0 & 0 \end{array}  \right]
\]
is \emph{not} positive semidefinite.

Entanglement of a bipartite Gaussian state can be established by an inequality on a certain nonlinear
function of determinants of $2\times 2$ blocks and a trace of a product of these blocks (see \cite{SimonR}, inequality (17)).
However, checking this condition implies the verification of various inequalities and, in the end, it essentially tantamounts
to the direct verification of the positivity of the roots of the characteristic polynomial. Therefore we follow this path.
As a bonus, we also see that, in our models, $\widetilde{\mathbf{S}}_{\text{\tiny red}}$ is not positive semidefinite
if and only if its determinant is strictly negative. In other words, the determinant of $\widetilde{\mathbf{S}}_{\text{\tiny red}}$
is an entanglement witness of the Gaussian bipartite state.

The characteristic polynomial of the above matrix is
\begin{eqnarray*}
& & \lambda^4
-  \frac{2 \widetilde{\beta}  \left(2(1-\kappa^2)-g^2 k^2\left(1+\kappa^2+g^2\right)\right)}
{\left(1-\kappa^2\right) \left(\left(1-\kappa^2\right)^2- g^4\kappa^2\right)} \lambda ^3 \\
& + & \frac{
 \widetilde{\beta}^2 (6 - (2 + g^2)^2 k^2 + (-2 + g^4) k^4)-2 (1 - k^2)^2 (1 - (2 + g^4) k^2 + k^4) }
{(1 - \kappa^2)^2 (1 - (2 + g^4) \kappa^2 + \kappa^4)} \lambda^2 \\
& - & \frac{   \widetilde{\beta}\left(2  \widetilde{\beta}^2 \left(2-g^2 k^2\right)+2 \left(1-k^2\right) \left(g^2 k^4+\left(g^4+g^2+2\right)
   k^2-2\right)\right)}{\left(1-k^2\right)^2 \left(1-\left(g^4+2\right) k^2+k^4\right)}\lambda \\
& + &  \frac{ \widetilde{\beta}^2 \left(g^4 \left(\kappa^2-\kappa^4\right)+2 \left(\kappa^4-1\right)\right)
 + \widetilde{\beta}^4+\left(1-\kappa^2\right)^2 \left((1-\kappa^2)^2-g^4\kappa^2\right)}
 {\left(1-\kappa^2\right)^2 \left((1-\kappa^2)^2-g^4\kappa^2\right)}
\end{eqnarray*}

The $\lambda^3$ coefficient is negative because it is minus the trace of $\mathbf{S}_{\text{\tiny red}}$ which has
positive diagonal entries or by \eqref{eq:cov-mat-pos-diag-11} and \eqref{eq:cov-mat-pos-diag-22}.
The $\lambda^2$ and $\lambda$ coefficients are positive by $ \widetilde{\beta}\geq 1$ and Lemma \ref{lem:1noise-lambdaaq-pos}.
Finally the zero order term (determinant) is  equal to
\[
 \frac{ \widetilde{\beta}^4- (1-\kappa^2)\left(2 +2 \kappa^2 -g^4 \kappa^2\right)\widetilde{\beta}^2
 + \left(1-\kappa^2\right)^2\left((1-\kappa^2)^2-g^4\kappa^2\right)}
 {\left(1-\kappa^2\right)^2 \left((1-\kappa^2)^2-g^4\kappa^2\right)}
\]
The numerator, as a function of $\widetilde{\beta}^2$ is a parabola with axis
the vertical line with abscissa $(1 - \kappa^2) (-1 - \kappa^2 + g^4  \kappa^2/2)$ which is negative by
\[
-1 - \kappa^2 + g^4  \kappa^2/2 < -1 - \kappa^2 +(1-\kappa^2)^2/2 = -1/2 -2\kappa^2 + \kappa^4/2 < -1/2
\]
as $|\kappa|<1$. Therefore the determinant can be negative for some values of $\widetilde{\beta}>1$
if and only if its value at $1$ is negative namely
\[
\kappa^2\left(-4(1 - \kappa^2)^2 + g^4\kappa^2 (1 - \kappa^2) + \kappa^6\right) < 0
\]
This inequality yields another upper bound on $g^2$. Finally,
solving the second order inequality in $\widetilde{\beta}$ that corresponds to negativity of the
determinant of $\widetilde{\mathbf{S}}_{\text{\tiny red}}$, we prove the following
\begin{theorem}\label{th:bgk}
The following conclusions hold
\begin{itemize}
\item[1.] There exists $\widetilde{\beta}>1$ for which the Gaussian invariant state is entangled if and only if
$g\not=0$, $0<|\kappa|<1$ and
\[
g^2 < \min\left\{\frac{(1-\kappa^2)}{|\kappa|}, \frac{\sqrt{\max\{4(1 - \kappa^2)^2-\kappa^6,0\}}}{|\kappa|\sqrt{1 - \kappa^2}} \right\}
\]
\item[2.] For such $g,\kappa$  the Gaussian invariant state is entangled  if and
only if $1 < \widetilde{\beta}<\widetilde{\beta}_*$ with
\begin{eqnarray*}
\widetilde{\beta}_*
&  = & 1-\kappa^4 +\frac{1}{2} \left(|\kappa|\left(1-\kappa^2\right)
\sqrt{16-g^4\kappa^2 \left(4-g^4\right) }-g^4 \kappa^2 \left(1-k^2\right)\right) \\
& = & \left(1-\kappa^2\right)\left(1+\kappa^2 +\frac{1}{2}  |\kappa |
\sqrt{16-g^4\kappa^2 \left(4-g^4\right) }-\frac{g^4 \kappa^2}{2}  \right)
\end{eqnarray*}
\end{itemize}
\end{theorem}

Figure \ref{fig:g2lt-funct-k} shows that the bound $g^2<(1 - \kappa^2)/|\kappa|$ for stability of $\mathbf{Z}$ prevails over the other one
for $|\kappa|< \sqrt{\left(5-\sqrt{13}\right)/2}\approx 0.8349996$... Moreover, for $|\kappa|>0.834996...$ one cannot get an entangled state
for any $g$ since the above upper bound goes to $0$.
\begin{figure}[h]
\begin{center}
\includegraphics[width=0.8\textwidth]{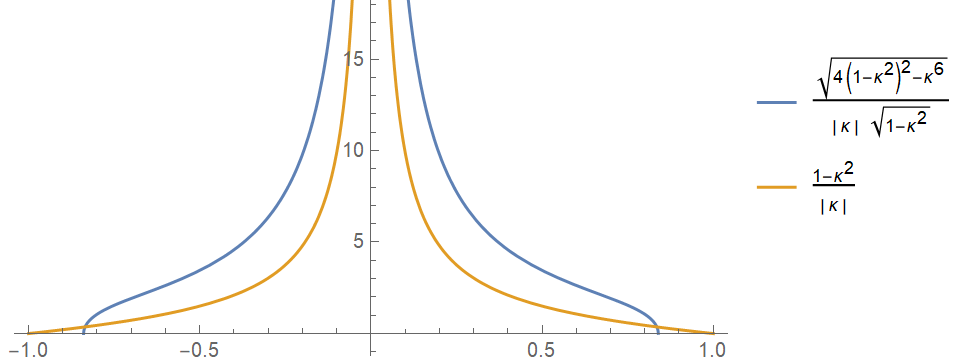}
\caption{Comparison of bounds on $g$ in terms of $\kappa$}\label{fig:g2lt-funct-k}
\end{center}
\end{figure}

\begin{corollary}\label{cor:1noise-ent}
If $0<|\kappa| < \sqrt{\left(5-\sqrt{13}\right)/2}$ one can always find a $g$ such that $g^2<(1 - \kappa^2)/|\kappa|$ and
$\widetilde{\beta} > 1$ for which the invariant state is entangled.
\end{corollary}


\noindent{\bf Remark.} It is interesting to compare, for fixed $\widetilde{\beta}$, values of $(\kappa,g)$
for which one gets stationary entangled states. Figure \ref{fig:1-noise-beta-comp}, shows the outcome
for $\widetilde{\beta}=1.05,$ $1.2$ and $1.28$.
\begin{figure}[h]
\begin{center}
\includegraphics[width=0.325\textwidth]{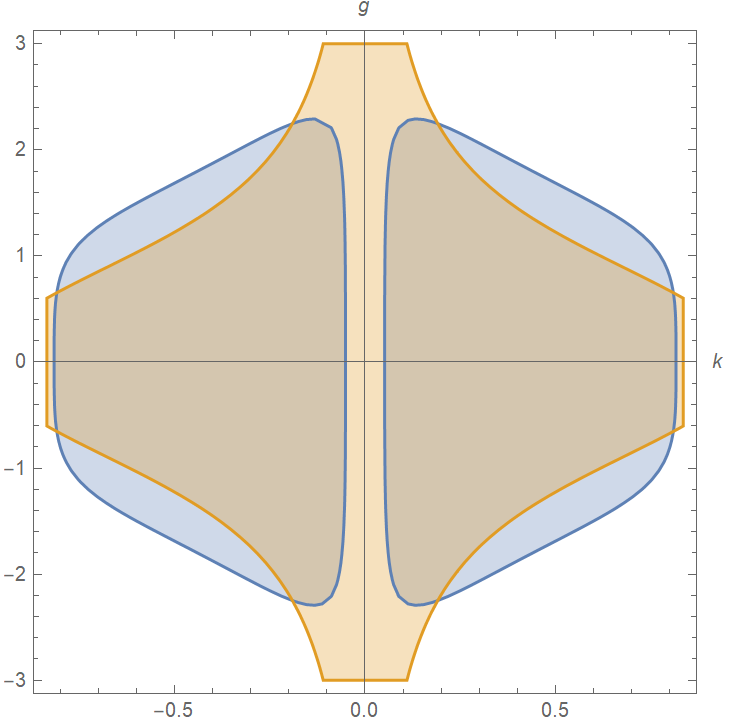} \
\includegraphics[width=0.32\textwidth]{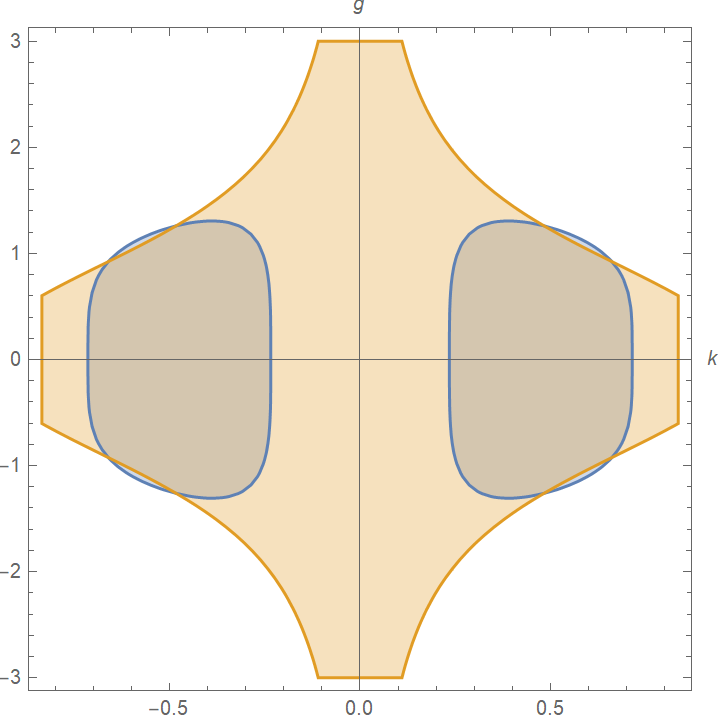} \
\includegraphics[width=0.32\textwidth]{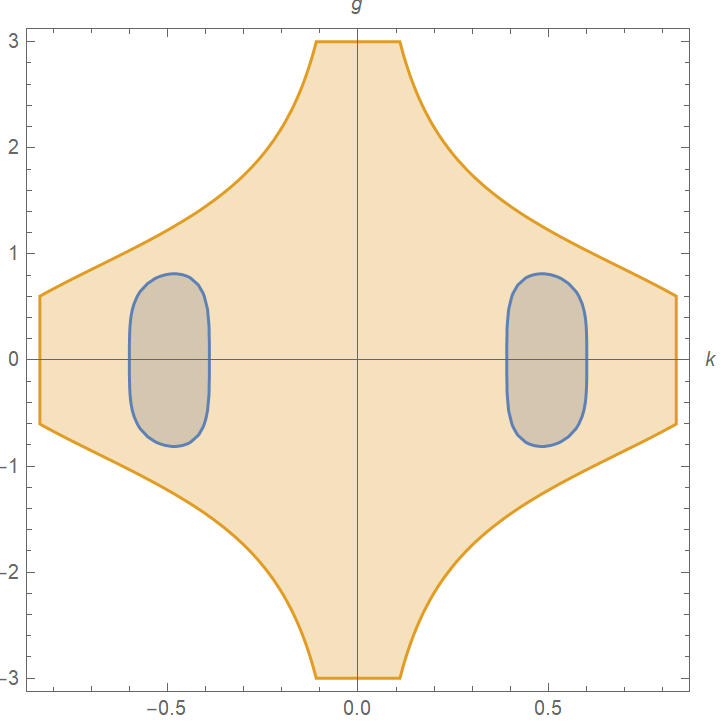} \\
  $\widetilde{\beta}=1.05$ \hskip 4truecm $\widetilde{\beta}=1.2$   \hskip 4truecm $\widetilde{\beta}=1.28$
  \caption{Comparison of entangled state regions, intersection of the light-orange and light-blue regions, in the $(\kappa,g)$ plane for different values of $\widetilde{\beta}$.} \label{fig:1-noise-beta-comp}
\end{center}
\end{figure}

The orange border region (light-orange color) is determined by the inequality $g^2|\kappa|<1-\kappa^2$,
the blue border region  (light blue color) by negativity
of the determinant of $\widetilde{\mathbf{S}}_{\text{\tiny red}}$. The Gaussian open system has a unique invariant
state which is Gaussian and entangled for $(\kappa,g)$ in the intersection of the two regions. One can note that, as the
temperature increases, there are less and less values of parameters $\kappa,g$ for which the
stationary state is entangled. The region shrinks to a single point, approximately $(0.51...,0)$, for $\widetilde{\beta}$
tending to $1.3$ and there is no Gaussian entangled state for $\widetilde{\beta}\gtrsim 1.3...$
It is worth noticing that, for any value of $\widetilde{\beta}$, the negative determinant has a minimum
for $\kappa\approx 0.5...$ and $g=0$. Indeed, quantifying the degree of entanglement by logarithmic negativity (see \cite{ViWe}
and the references therein), one could see that entanglement is bigger for $k\approx 0.5$, $g$ near $0$ (but non-zero)
and $\widetilde{\beta}$ slightly bigger than $1$. However, we do not develop these considerations in detail because they
only confirm the results already found without adding much.

\smallskip

\noindent{\bf Remark.} It is worth noticing that, choosing a $\widetilde{\beta}$ depending on $\kappa$, one
can get an entangled stationary state for essentially all parameters $(\kappa,g)$ fulfilling $g^2|\kappa|<1-\kappa^2$.
Indeed, the Taylor expansion of $\widetilde{\beta}_*$ in the variable $g^4$ at $g=0$ yields
$\widetilde{\beta}_* = 1 + 2 |\kappa| - 2 |\kappa|^3 - \kappa^4 + o(g^4)$. This suggests a choice of $\widetilde{\beta}<\widetilde{\beta}_*$
like $\widetilde{\beta}=1 + |\kappa|(1 -|\kappa|)/4$. Figure \ref{fig:ent-b(k)-kg} below shows the region
(again, intersection of light blue and yellow gray regions)
in the stripe $\{\,(\kappa,g)\,\mid\, 0<|\kappa|<0.8\,\}$ of values $(\kappa,g)$ for which this choice of
$\widetilde{\beta}$ yields an entangled state. Warning: points $(0,g)$ ($g$ arbitrary) must be removed.
\begin{figure}[h]
\begin{center}
\includegraphics[width=0.36\textwidth]{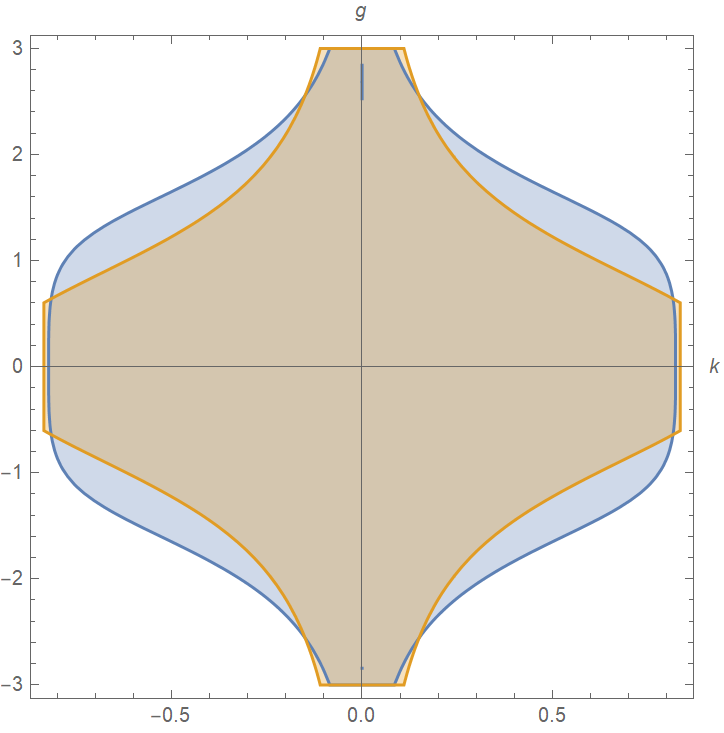}
\caption{$\widetilde{\beta}=1 + |\kappa|(1 -|\kappa|)/4$, determinant of $\widetilde{\mathbf{S}}_{\text{\tiny red}}$
negative in the light blue coloured region including almost all the stability region.}\label{fig:ent-b(k)-kg}
\end{center}
\end{figure}

Finally, for comparison with the two noise model that we shall analyze in the next section it is useful to have a
look at Figure \ref{fig:1-noise-kappa-comp}.
\begin{figure}[h]
\begin{center}
\includegraphics[width=0.27\textwidth]{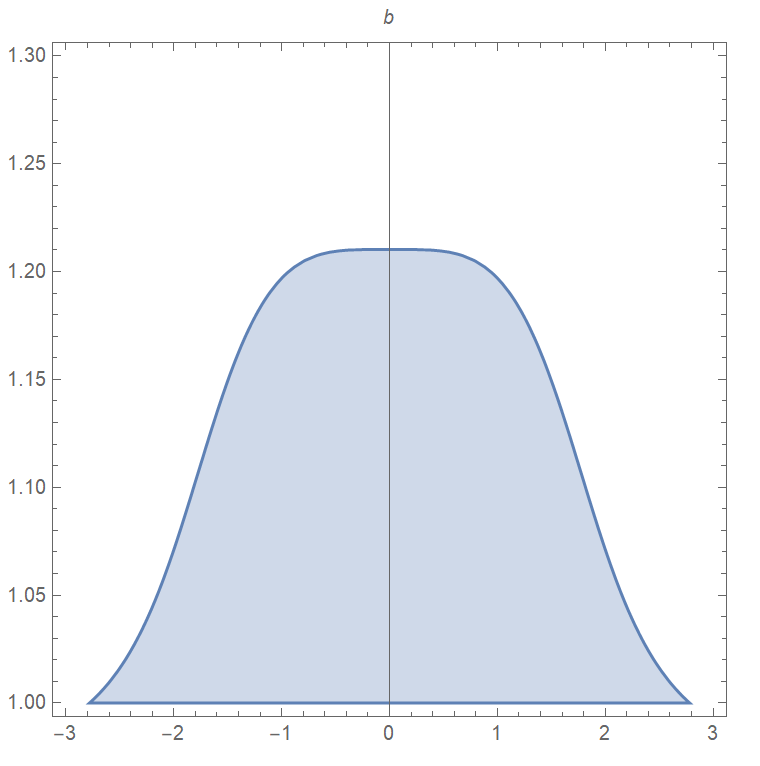} \
\includegraphics[width=0.28\textwidth]{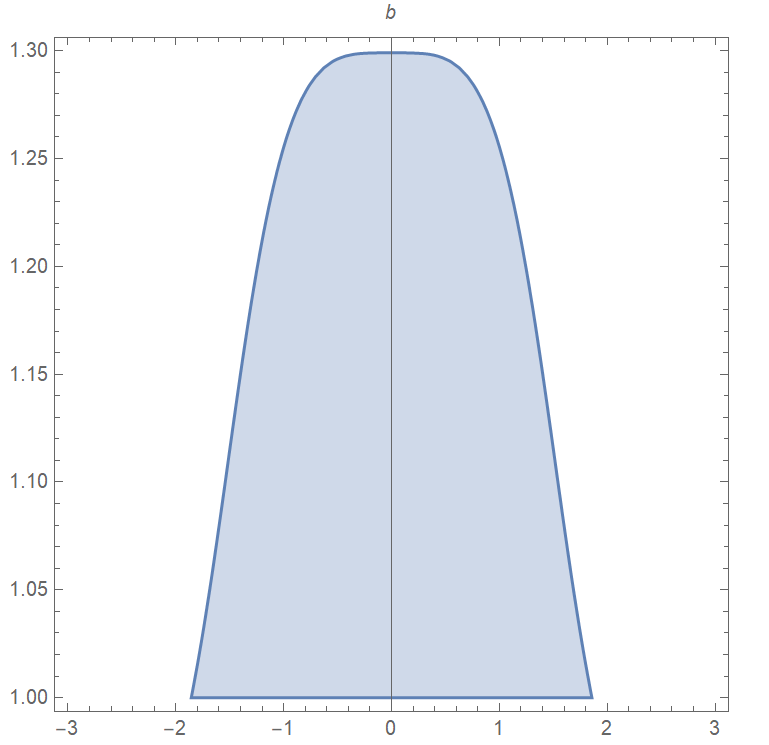} \
\includegraphics[width=0.28\textwidth]{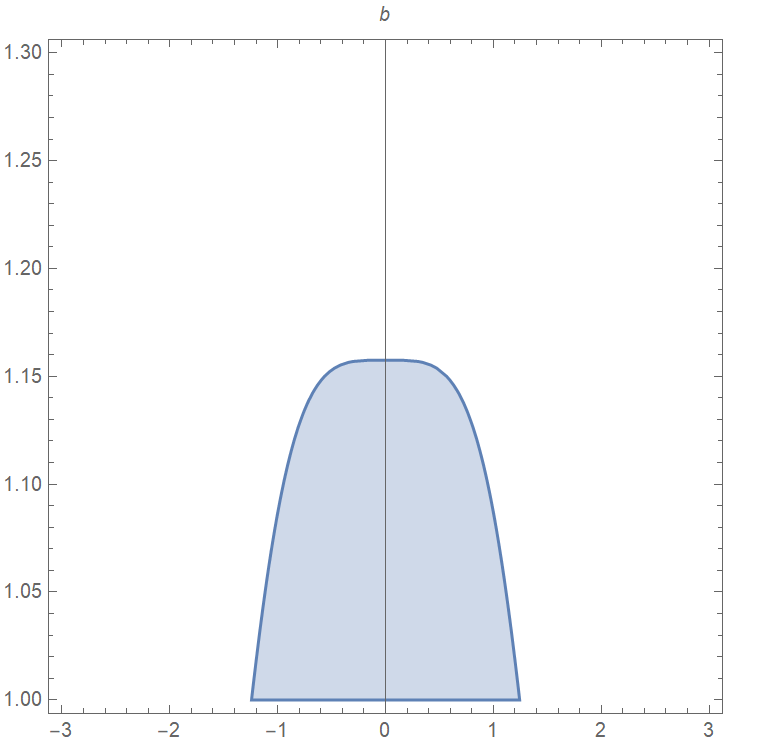} \\
  $\kappa=1/4$ \hskip 3.2truecm $\kappa=1/2$   \hskip 3.2truecm $\kappa=3/4$
  \caption{Entangled state for $(g,b=\widetilde{\beta})$ in the blue region for different $\kappa$ } \label{fig:1-noise-kappa-comp}
\end{center}
\end{figure}

\section{Two noise model}\label{sect:two-noises-model}

In this section we analyze the situation when also the second party interacts with a noise labelled by the
index $3$ and modelled by another damped and pumped quantum harmonic oscillator with inverse
temperature $\beta_3>0$. The inverse temperature of the other reservoir will be denoted by $\beta_0>0$.
Consider the 4-mode Fock space $\Gamma(\mathbb{C}^4)$ and the Gaussian GKLS generator \eqref{eq:GKLS}
with $m=4$, $L_1,L_2$ as in \eqref{eq:1noiseL1L2} and
\begin{eqnarray}
 & & \qquad L_3 = \left(\frac{\mathrm{e}^{\beta_3}}{\mathrm{e}^{\beta_3}-1}\right)^{1/2} a_3
 \qquad
  L_4 = \left(\frac{1}{\mathrm{e}^{\beta_3}-1}\right)^{1/2} a^\dagger_3  \label{eq:2noiseL3L4}  \\
  H & = & H_S + H_{\text{\tiny int} } \\
& = & \frac{1}{2}  \left(a_1 a^\dagger_{2}+a^\dagger_1 a_{2}\right) +\frac{\kappa}{2}\sum_{j=1}^2\left( a_j^2 + a_j^{\dagger 2}\right)
+\frac{g}{2}\left(a_0a^\dagger_{1}+a^\dagger_0 a_{1}+a_2a^\dagger_{3}+a^\dagger_2 a_{3}\right) \label{eq:2noiseH}
 \end{eqnarray}
where $g\not=0$ so that the interaction between parties 0-1 and 2-3 is non-zero. The Hilbert space of
the whole open quantum system is
\begin{eqnarray*}
  \mathsf{h}=   \Gamma(\mathbb{C}^4) & = & \underbrace{ {\Gamma({\mathbb{C}})}}_{\text{reservoir}}\otimes
   \underbrace{\Gamma({\mathbb{C}})\otimes  \Gamma({\mathbb{C}})}_{\text{system}}
                                    \otimes \underbrace{ {\Gamma({\mathbb{C}})}}_{\text{reservoir}} \\
    {\text{\footnotesize (indexes)}}
   & &  {\quad \text{\footnotesize 0}\hskip 1.4truecm (\text{\footnotesize 1} \hskip 0.2truecm \otimes \hskip 0.2truecm
   \text{\footnotesize 2})\qquad \qquad \text{\footnotesize 3}}
\end{eqnarray*}

The matrices ${\mathbf{Z}},{\mathbf{C}}$  in the explicit formula for the
action of the Gaussian QMS on Weyl operators (see \cite{AFPMJM} Theorem 2.4) are as follows.
\begin{eqnarray}
\mathbf{Z} & = & \frac{1}{2}\left[ \begin{array}{cccccccc}
                      -1 & 0 & 0 & 0 & 0 & -g & 0 & 0 \\
                      0 & 0 & 0 & 0  & -g & \kappa & -1 & 0 \\
                      0 & 0 & 0 & 0  & 0 & -1 & \kappa & -g \\
                     0 & 0 & 0 & -1  & 0 & 0 & -g & 0 \\
                     0 & g & 0 & 0 & -1 & 0 & 0 & 0 \\
                     g & \kappa & 1 & 0 & 0 & 0 & 0 & 0 \\
                     0 & 1 & \kappa& g & 0 & 0 & 0 & 0 \\
                     0 & 0 & g & 0 &0 & 0 & 0 & -1 \\
                    \end{array}\right]                          \label{eq:Z-2noise}          \\
\mathbf{C} & = &  \left[ \begin{array}{cccccccc}
                      \widetilde{\beta}_0 & 0 & 0 & 0 & 0 & 0 & 0 & 0 \\
                      0 & 0 & 0 & 0  & 0 & 0 & 0 & 0 \\
                      0 & 0 & 0 & 0  & 0 & 0 & 0 & 0 \\
                     0 & 0 & 0 & \widetilde{\beta}_3  & 0 & 0 & 0 & 0 \\
                     0 & 0 & 0 & 0 & \widetilde{\beta}_0 & 0 & 0 & 0 \\
                     0 & 0 & 0 & 0 & 0 & 0 & 0 & 0 \\
                     0 &0 & 0 & 0 & 0 & 0 & 0 & 0 \\
                     0 & 0 & 0 & 0 &0 & 0 & 0 & \widetilde{\beta}_3 \\
                    \end{array}\right]  \label{eq:C-2noise}
\end{eqnarray}
where $\widetilde{\beta}_j=\coth(\beta_j/2)$ for $j=0,3$.
The characteristic polynomial of $\mathbf{Z}$ is
\begin{equation}\label{eq:char_poly_Z}
\left(\lambda ^4+\lambda^3 +  \frac{2(1+g^2)-\kappa^2}{4} \lambda ^2
+ \frac{1-\kappa^2+g^2}{4}\lambda
+\frac{1-\kappa^2+g^4}{16}\right)^2
\end{equation}
Recall that, given a fourth order polynomial,
\[
\lambda^4+a_3\lambda^3+a_2\lambda^2+a_1\lambda+a_0
\]
the Routh-Hurwitz criterion for stability (see  \cite{Math24}) reads as follows:
\[
a_i>0\ (i=0,1,2,3),\quad a_3a_2>a_1,\quad a_3a_2a_1>a_3^2a_0+a_1^2
\]
The first condition is equivalent to $\kappa^2 <\min\{\,1+g^2,1+g^4\,\}$ 
the second condition always hold and the last one reduces to  $g^2(2-\kappa^2)>0$. Summarizing
\begin{lemma}
Assume $g\not=0$. The matrix $\mathbf Z$ \eqref{eq:Z-2noise}  is stable if and only if
\begin{equation}\label{eq:stability-Z-2noises}
\kappa^2<\min\{\,1+g^4,2\,\}.
\end{equation}
\end{lemma}

The covariance matrix of the invariant state is the unique solution of
$\mathbf{Z}^{\rm\scriptstyle T} \mathbf{S}+\mathbf{S}\mathbf{Z}+\mathbf{C}=0$
by Theorem \ref{thm:esZCesZ-conv}. However, by leaving arbitrary all the parameters, we get a cumbersome expression.
Therefore we begin by the case $\kappa=0$ in which we get the following
\begin{proposition}\label{prop:2noises-k=0}
If $\kappa=0$ the covariance matrix of the unique invariant state of the bipartite system is
\[
\mathbf{S}_{\text{\rm\tiny red}} =
\frac{\widetilde{\beta}_0+\widetilde{\beta}_3}{2} \unit_4
+ \frac{(\widetilde{\beta}_0-\widetilde{\beta}_3)g^2}{4(1+g^4)}
\left[ \begin{array}{cccc}
-(1-g^2)  & 0 & 0 & -(1+g^2) \\
0 & 1 - g^2 & 1 + g^2 & 0 \\
0 & 1 + g^2 & -(1 - g^2) & 0 \\
-(1+g^2) & 0 & 0 & (1- g^2)
 \end{array} \right]
\]
The invariant state is separable for all $g\not=0,\widetilde{\beta}_0>1,\widetilde{\beta}_3>1$.
\end{proposition}

\begin{proof}
The characteristic polynomial of the matrix $\widetilde{\mathbf{S}}_{\text{\tiny red}}$
defined as in \eqref{eq:S-tilde}  is
 \begin{eqnarray*}
 & & \lambda^4 -2(\widetilde{\beta}_0+\widetilde{\beta}_3)\lambda^3
 + \left(\frac{3}{2}(\widetilde{\beta}_0+\widetilde{\beta}_3)^2-2-\frac{(\widetilde{\beta}_0-\widetilde{\beta}_3)^2g^4}{4(1+g^4)}\right)
 \lambda^2 \\
 & &
 +(\widetilde{\beta}_0+\widetilde{\beta}_3)\left(\frac{(\widetilde{\beta}_0-\widetilde{\beta}_3)^2 g^4}{4 \left(g^4+1\right)}
 -\frac{1}{2} (\widetilde{\beta}_0+\widetilde{\beta}_3)^2+2\right)\lambda  +1 -\frac{1}{2} (\widetilde{\beta}_0+\widetilde{\beta}_3)^2
 +\frac{1}{16} (\widetilde{\beta}_0+\widetilde{\beta}_3)^4 \\
 & &
   -\frac{g^4 (\widetilde{\beta}_0-\widetilde{\beta}_3)^2 (\widetilde{\beta}_0+\widetilde{\beta}_3)^2}{16 \left(1+g^4\right)}
   +\frac{g^8 (\widetilde{\beta}_0-\widetilde{\beta}_3)^4+32
   g^6 (\widetilde{\beta}_0-\widetilde{\beta}_3)^2}{64 \left(1+g^4\right)^2}
 \end{eqnarray*}
where the last term (determinant of $\widetilde{\mathbf{S}}_{\text{\tiny red}}$) can also be written as
\[
\left(\left(\frac{\widetilde{\beta}_0+\widetilde{\beta}_3}{2}\right)^2  -1\right)^2
+ \frac{(\widetilde{\beta}_0-\widetilde{\beta}_3)^2g^4}{16\left(1+g^4\right)}
\left(\frac{g^4 (\widetilde{\beta}_0-\widetilde{\beta}_3)^2+32
   g^2 }{4\left(1+g^4\right) }- (\widetilde{\beta}_0+\widetilde{\beta}_3)^2 \right)
\]
Note that the $\lambda^2$ coefficient satisfies
\begin{eqnarray*}
\frac{3}{2}(\widetilde{\beta}_0+\widetilde{\beta}_3)^2-2-\frac{(\widetilde{\beta}_0-\widetilde{\beta}_3)^2g^4}{4(1+g^4)}
& = & \left(\frac{(\widetilde{\beta}_0+\widetilde{\beta}_3)^2}{2}-2\right)
+(\widetilde{\beta}_0+\widetilde{\beta}_3)^2 -\frac{(\widetilde{\beta}_0-\widetilde{\beta}_3)^2g^4}{4(1+g^4)} \\
& \geq & (\widetilde{\beta}_0+\widetilde{\beta}_3)^2 -\frac{(\widetilde{\beta}_0-\widetilde{\beta}_3)^2}{4} \geq 0
\end{eqnarray*}
Minus the first order coefficient is a multiple of
\begin{eqnarray*}
\frac{1}{2} (\widetilde{\beta}_0+\widetilde{\beta}_3)^2-2-\frac{(\widetilde{\beta}_0-\widetilde{\beta}_3)^2 g^4}{4 \left(g^4+1\right)}
& \geq &  \frac{1}{2} (\widetilde{\beta}_0+\widetilde{\beta}_3)^2-2-\frac{(\widetilde{\beta}_0-\widetilde{\beta}_3)^2 }{4 }  \\
& = &  2\left( \left(\frac{ \widetilde{\beta}_0+\widetilde{\beta}_3 }{2}\right)^2-1  \right)
-\frac{(\widetilde{\beta}_0-\widetilde{\beta}_3)^2 }{4 }  \\
& \geq &  \left(\frac{ \widetilde{\beta}_0+\widetilde{\beta}_3 }{2}\right)^2-1
-\left(\frac{\widetilde{\beta}_0-\widetilde{\beta}_3  }{2 }\right)^2 \\
& = & \widetilde{\beta}_0 \widetilde{\beta}_3 -1
\end{eqnarray*}
Therefore the first order coefficient is strictly negative.

The $0$ order term of the fourth order polynomial is bigger than
\begin{eqnarray*}
   & &  \left(\left(\frac{\widetilde{\beta}_0+\widetilde{\beta}_3}{2}\right)^2  -1\right)^2
-\frac{(\widetilde{\beta}_0-\widetilde{\beta}_3)^2 (\widetilde{\beta}_0+\widetilde{\beta}_3)^2 g^4}{16\left(1+g^4\right)}\\
   &\geq  & \left(\left(\frac{\widetilde{\beta}_0+\widetilde{\beta}_3}{2}\right)^2  -1\right)^2
-\left(\frac{\widetilde{\beta}_0-\widetilde{\beta}_3}{2}\right)^2 \left(\frac{\widetilde{\beta}_0+\widetilde{\beta}_3}{2}\right)^2 \\
& = & \left(  \frac{\widetilde{\beta}_0+\widetilde{\beta}_3}{2}\widetilde{\beta}_0-1\right)
\left(  \frac{\widetilde{\beta}_0+\widetilde{\beta}_3}{2}\widetilde{\beta}_3-1\right) >0
\end{eqnarray*}
Therefore by Descartes' rule of signs all roots are strictly positive and the invariant state is not entangled.
\end{proof}

\begin{proposition}\label{prop:2noises-k=1}
If $\kappa=1$, defining $\widehat{g}=(1+g^2)/g^4$, the covariance matrix $\mathbf{S}_{\text{\rm\tiny red}} $ of the unique invariant state is
\[
\frac{1}{2}\left[
\begin{array}{cccc}
  \scriptstyle{\left(2+\widehat{g}\right)(\widetilde{\beta}_0+\widetilde{\beta}_3)-(\widetilde{\beta}_3-\widetilde{\beta}_0)}
   & \scriptstyle{\left(1+\widehat{g}\right) (\widetilde{\beta}_0+\widetilde{\beta}_3)}
   & \scriptstyle{2g^2\widehat{g}\,\widetilde{\beta}_0} & \scriptstyle{-g^2 \widehat{g} (\widetilde{\beta}_0-\widetilde{\beta}_3)} \\
 \scriptstyle{\left(1+\widehat{g}\right)  (\widetilde{\beta}_0+\widetilde{\beta}_3)}
   & \scriptstyle{\left(2+\widehat{g}\right)(\widetilde{\beta}_0+\widetilde{\beta}_3)}
   \scriptstyle{+(\widetilde{\beta}_3-\widetilde{\beta}_0)} & \scriptstyle{g^2\widehat{g}\,(\widetilde{\beta}_0-\widetilde{\beta}_3)}
   &  \scriptstyle{ 2g^2\widehat{g}\,\widetilde{\beta}_3 } \\
  \scriptstyle{2g^2\widehat{g}\,\widetilde{\beta}_0} &
    \scriptstyle{g^2\widehat{g}\, (\widetilde{\beta}_0-\widetilde{\beta}_3)}
    & \scriptstyle{\left(2+\widehat{g}\right)(\widetilde{\beta}_0+\widetilde{\beta}_3)-(\widetilde{\beta}_3-\widetilde{\beta}_0)}
    & \scriptstyle{-\left(1+\widehat{g}\right)(\widetilde{\beta}_0+\widetilde{\beta}_3)} \\
 \scriptstyle{-g^2 \widehat{g} (\widetilde{\beta}_0-\widetilde{\beta}_3)} & \scriptstyle{2g^2\widehat{g}\,\widetilde{\beta}_3}
 & \scriptstyle{-\left(1+\widehat{g}\right)(\widetilde{\beta}_0+\widetilde{\beta}_3)}
 & \scriptstyle{\left(2+\widehat{g}\right)(\widetilde{\beta}_0+\widetilde{\beta}_3)}
   \scriptstyle{+(\widetilde{\beta}_3-\widetilde{\beta}_0)} \\
\end{array}
\right]
\]
For all $\widetilde{\beta}_0,\widetilde{\beta}_3 > 1$ such that
  \begin{equation}\label{eq:2-noise-big-g-neg-det}
  2(\widetilde{\beta}_0 + \widetilde{\beta}_3 -1 )^{2}  - (\widetilde{\beta}_0 - \widetilde{\beta}_3)^2 < 6,
  \end{equation}
  one can find a $g^2$ big enough for which the stationary state is entangled.
\end{proposition}

\begin{proof}
The characteristic polynomial of $\widetilde{\mathbf{S}}_{\text{\tiny red}}$ now is
\begin{eqnarray*}
  & & \lambda^4 -\frac{2(\widetilde{\beta}_0+\widetilde{\beta}_3)(1+g^2+2g^4)}{g^4}\lambda^3 \\
  & & +\left( \frac{\left(7 g^8+4 g^6+9 g^4+4 g^2+2\right) (\widetilde{\beta}_0+\widetilde{\beta}_3)^2}{2 g^8}
  +\frac{2 \widetilde{\beta}_0 \widetilde{\beta}_3 \left(3 g^4+4    g^2+2\right)}{g^4}-2 \right)\lambda^2 \\
   & & - \frac{1+g^2+2g^4}{2g^8}\left( (\widetilde{\beta}_0 +\widetilde{\beta}_3)^3(1+g^4)
   +4g^4(\widetilde{\beta}_0 \widetilde{\beta}_3-1)(\widetilde{\beta}_0 +\widetilde{\beta}_3)\right) \lambda \\
   & & +\frac{ (\widetilde{\beta}_0+\widetilde{\beta}_3)^4\left(1+g^4\right)^2}{16 g^8}
   -\frac{\left(2+4g^2+5 g^4+3 g^8\right) (\widetilde{\beta}_0+\widetilde{\beta}_3)^2}{2g^8} \\
   & & +\frac{\widetilde{\beta}_0\widetilde{\beta}_3 (\widetilde{\beta}_0+\widetilde{\beta}_3)^2\left(1+g^4\right)}{2 g^4}
   -\frac{2 \widetilde{\beta}_0\widetilde{\beta}_3 \left(2+4 g^2+g^4\right)}{g^4}+1+\widetilde{\beta}_0^2\widetilde{\beta}_3^2
\end{eqnarray*}
The $\lambda^3$ coefficient is clearly negative. The $\lambda^2$ coefficient is positive by
\begin{eqnarray*}
  & & \frac{\left(7 g^8+4 g^6+9 g^4+4 g^2+2\right) (\widetilde{\beta}_0+\widetilde{\beta}_3)^2}{2 g^8}
  +\frac{2 \widetilde{\beta}_0 \widetilde{\beta}_3 \left(3 g^4+4g^2+2\right)}{g^4}-2  \\
  & = & \frac{7 (\widetilde{\beta}_0+\widetilde{\beta}_3)^2}{2} - 2
  +\frac{\left(4 g^6+9 g^4+4 g^2+2\right) (\widetilde{\beta}_0+\widetilde{\beta}_3)^2}{2 g^8}
  +\frac{2 \widetilde{\beta}_0 \widetilde{\beta}_3 \left(3 g^4+4    g^2+2\right)}{g^4} \\
  & \geq & 12  +\frac{\left(4 g^6+9 g^4+4 g^2+2\right) (\widetilde{\beta}_0+\widetilde{\beta}_3)^2}{2 g^8}
  +\frac{2 \widetilde{\beta}_0 \widetilde{\beta}_3 \left(3 g^4+4    g^2+2\right)}{g^4}  > 0
\end{eqnarray*}
The $\lambda$ coefficient is clearly negative.

The coefficient of $\lambda^0$ (i.e. the determinant) is more complicated. For this reason it has been
written collecting sums and products of inverse temperatures.
Computing the limit as $g^2$ goes to infinity of the determinant we find
\begin{eqnarray*}
& & 1+ \frac{(\widetilde{\beta}_0 +\widetilde{\beta}_3)^4}{16}
+ \frac{\widetilde{\beta}_0 \widetilde{\beta}_3 (\widetilde{\beta}_0 +\widetilde{\beta}_3)^2}{2} +   \widetilde{\beta}_0^2 \widetilde{\beta}_3^2
-\frac{3}{2}   (\widetilde{\beta}_0 +\widetilde{\beta}_3)^2-2 \widetilde{\beta}_0 \widetilde{\beta}_3 \\
& = & 1+\left(\frac{(\widetilde{\beta}_0 +\widetilde{\beta}_3)^2}{4} +\widetilde{\beta}_0  \widetilde{\beta}_3 \right)^2
-\frac{3}{2}   (\widetilde{\beta}_0 +\widetilde{\beta}_3)^2-2 \widetilde{\beta}_0 \widetilde{\beta}_3 \\
& = & \left(\frac{(\widetilde{\beta}_0 +\widetilde{\beta}_3)^2}{4} +\widetilde{\beta}_0  \widetilde{\beta}_3 -1\right)^2
- (\widetilde{\beta}_0 +\widetilde{\beta}_3)^2.
\end{eqnarray*}
Since $\widetilde{\beta}_0,  \widetilde{\beta}_3 >1$, the argument of the squares are positive and the above limit is negative
if and only if $(\widetilde{\beta}_0 +\widetilde{\beta}_3)^2  +4\widetilde{\beta}_0  \widetilde{\beta}_3 -4 < 4(\widetilde{\beta}_0 +\widetilde{\beta}_3)$. A little manipulation shows that this is equivalent to inequality \eqref{eq:2-noise-big-g-neg-det}.
\end{proof}

\smallskip

\noindent{\bf Remark.}
The left-hand side of \eqref{eq:2-noise-big-g-neg-det} determines an hyperbola with center $(1/2,1/2)$ and asymptotes
$(\sqrt{2}\mp 1)\widetilde{\beta}_0+(\sqrt{2}\pm 1)\widetilde{\beta}_3=\sqrt{2}$ intersecting the half-lines
$\widetilde{\beta}_0=1, \widetilde{\beta}_3 > 1$
and $\widetilde{\beta}_0 > 1, \widetilde{\beta}_3= 1$ at points $(1,2\sqrt{2}-1)$ and $(2\sqrt{2}-1,1)$.
If $\widetilde{\beta}_0,  \widetilde{\beta}_3 > 1$ lie under the hyperbola
then one can always find a $g^2>0$ for which the stationary state is entangled. A closer look at the inequality above also shows
that one can allow a simultaneous slight increase of $\widetilde{\beta}_0$, $\widetilde{\beta}_3$ together with their difference
still obtaining an entangled steady state. Roughly speaking, larger temperature differences allow the use of hotter thermal baths,
but not too hot due to the upper bounds $\widetilde{\beta}_0,\widetilde{\beta}_3\leq 2\sqrt{2}-1$. \\
Moreover, it is clear that if $|\widetilde{\beta}_0 - \widetilde{\beta}_3|$ increases, with constant $\widetilde{\beta}_0 + \widetilde{\beta}_3$
and $g$, the minimum eigenvalue of $\widetilde{\mathbf{S}}_{\text{\tiny red}}$ becomes more negative. Therefore, one could say that
the steady state becomes ``more entangled''.
\medskip

The temperature difference, however, seems less relevant than the other parameters, therefore we now set
$\widetilde{\beta}_0 = \widetilde{\beta}_3=b$ and analyze the dependence on the two variables $g,b$.
The matrix $\widetilde{\mathbf{S}}_{\text{\tiny red}}$  becomes
\[
\widetilde{\mathbf{S}}_{\text{\tiny red}}=\left[\begin{matrix}
b \frac{ 1 + g^2+2g^4}{2g^4}& b \frac{1+g^2+g^4}{2g^4}&b\frac{1+g^2+g^4}{2g^4}-\mi&0\\
b\frac{1+g^2+g^4}{2g^4}&b\frac{1+g^2+2g^4}{2g^4}&0&b\frac{1+g^2}{2g^2}+\mi\\
b\frac{1+g^2+g^4}{2g^4}+\mi&0&b\frac{1+g^2+2g^4}{2g^4}&-b\frac{1+g^2+g^4}{2g^4}\\
0&b\frac{1+g^2}{2g^2}-\mi&-b\frac{1+g^2+g^4}{2g^4}&b\frac{1+g^2+2g^4}{2g^4}\end{matrix}\right].
\]
The determinant is now
\[
\left(b^4 \left(2 g^4+1\right)^2-2 b^2 \left(4 g^8+4 g^6+7 g^4+4 g^2+2\right)+g^8\right)/g^8
\]
and it can be factored as follows
\[
g^{-8} \left(\left(2 b^2+2 b-1\right) g^4+2 b g^2+b^2+2 b\right)\left(\left(2 b^2-2 b-1\right) g^4-2 b g^2+b^2-2 b\right)
\]
The first two factors are always positive. One can see that the third factor is negative if and only if
\[
1\leq b < \frac{1+g^2+g^4+\sqrt{(1+g^2+g^4)^2+g^4(1+2g^4)}}{1+2 g^4}.
\]
Figure \ref{fig:2noises-entgl-g-b.png} shows the region in the $(g,b)$ plane in which this inequality holds.
\begin{figure}[h]
\begin{center}
\includegraphics[width=0.64\textwidth]{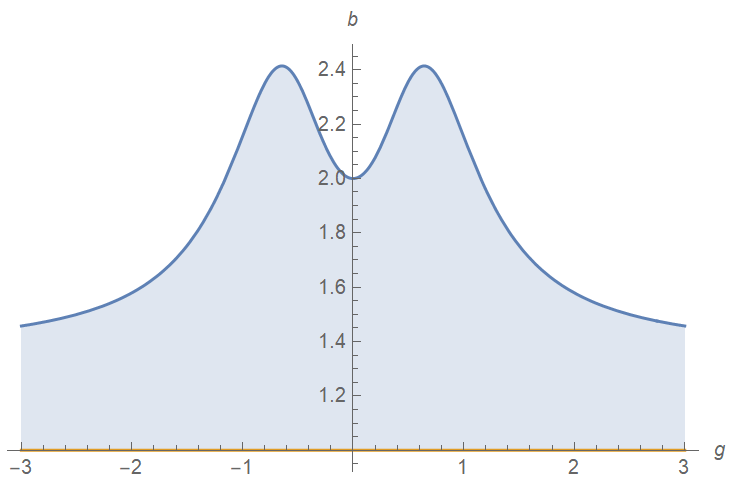}
\caption{$\kappa=1$, two noises, entangled state for $(g,b)$ in the shaded region. Comparison with Figure \ref{fig:1-noise-kappa-comp} shows that
adding noise ``improves'' entanglement.}\label{fig:2noises-entgl-g-b.png}
\end{center}
\end{figure}

Since
\begin{eqnarray*}
& & \lim_{|g|\to+\infty}\frac{1+g^2+g^4+\sqrt{(1+g^2+g^4)^2+g^4(1+2g^4)}}{1+2 g^4}= \frac{1+\sqrt{3}}{2}, \\
& & \frac{1+g^2+g^4+\sqrt{(1+g^2+g^4)^2+g^4(1+2g^4)}}{1+2 g^4}= 2 \quad \Leftrightarrow \quad g=0\ \text{or}\ g=\pm\frac{2}{\sqrt{3}}
\end{eqnarray*}
one can draw the following conclusion:

\begin{theorem}\label{th:2noise-ent}
In the two noise model with $\kappa=1$ and $\widetilde{\beta}_0=\widetilde{\beta}_3=b$ entangled stationary states arise
\begin{itemize}
\item[1.] For $1 < b \leq (1+\sqrt{3})/2$ (low temperature) and any coupling constant $g$,
\item[2.] For $(1+\sqrt{3})/2 < b \leq 2$ (``moderate/low'' temperature) for small coupling strength
\[
0<g^2<\frac{b+\sqrt{2b(b-1)(-b^2+2b+1)}}{2b^2-2b-1}
\]
\item[3.] For $2 < b < 1+\sqrt{2}$ (``moderate/high'' temperature) for small, but not too small, coupling strength,
\[
\frac{b-\sqrt{2b(b-1)(-b^2+2b+1)}}{2b^2-2b-1}<g^2<\frac{b+\sqrt{2b(b-1)(-b^2+2b+1)}}{2b^2-2b-1}
\]
\item[4.] Never for $b \geq  1+\sqrt{2}$ (high temperature).
\end{itemize}
\end{theorem}

Note that the bound $b \leq (1+\sqrt{3})/2$ also comes from \eqref{eq:2-noise-big-g-neg-det} setting
$\widetilde{\beta}_0=\widetilde{\beta}_3$.

\section{Conclusion and outlook}

We proved that the bipartite open system with Hamiltonian \eqref{eq:HS} and a local interaction with
one or two reservoirs has a unique Gaussian entangled stationary state and that any initial state converges towards
this stationary state under some explicit conditions on the interaction strength and reservoir temperature.
In this way, the created quantum entanglement turns out to be stable and survives for all
times also for non vanishing bath temperatures even with external noise.
We also mentioned entanglement quantification via logarithmic negativity, although we expect that it will just
confirm predictions of determinant negativity of $\widetilde{\mathbf{S}}_{\text{\tiny red}}$. \\
In our model parties $1$ and $2$ directly interact. However, it is not difficult to make them non interacting by the unitary
transformation (a passive element, in the language of \cite{WoEiPl})
\[
U = \left[\begin{array}{cccc} 1/\sqrt{2} & 1/\sqrt{2} & 0 & 0 \\
                             -1/\sqrt{2} & 1/\sqrt{2} & 0 & 0 \\
                             0 & 0 & 1/\sqrt{2} & 1/\sqrt{2} \\
                             0 & 0 & -1/\sqrt{2} & 1/\sqrt{2}
\end{array} \right]
\]
In this way, with respect to the transformed annihilation and creation operators $H_S$,
$b_1=(a_1+a_2)/\sqrt{2}, b_2=(-a_1+a_2)/\sqrt{2}$ becomes
\[
\frac{\omega}{2}\left(b_1^\dagger b_1 - b_2^\dagger b_2\right)
+ \frac{\kappa}{2}\sum_{j=1}^2\left(b_j^{\dagger\, 2} + b_j^2\right)
\]
and one finds that coupling with external baths can induce entanglement even if the system parties are not interacting. \\
It might be interesting to analyze the time dependence of entanglement. If conditions on parameters hold, a separable initial
state will become entangled in finite time because the determinant of $\widetilde{\mathbf{S}}_{\text{\tiny red}}$ is
a continuous function of time by \eqref{eq:invariantParameterEvolution} and its limit at infinity is strictly negative.
The quantification of entanglement sudden birth as well as the time behaviour of entanglement seem much more involved as
they depends on symplectic diagonalization of a time dependent family of covariance matrices.

\section*{Appendix A. Solving $\mathbf{Z}^T\mathbf{S}+\mathbf{S}\mathbf{Z}+\mathbf{C}=0$}

We consider the two noises case for simplicity. Write
\[
\mathbf{C} =\left[ \begin{array}{cccccc}  \widetilde{\beta}_0 & \dots & 0 & 0 &\dots & 0 \\
                                          \dots      & \dots & \dots & \dots & \dots & \dots \\
                                         0      & \dots & \widetilde{\beta}_3 & 0 & \dots & 0 \\
                                         0      & \dots & 0 & \widetilde{\beta}_0 & \dots & 0 \\
                                         \dots      & \dots & \dots & \dots & \dots & \dots \\
                                         0      & \dots & 0 & 0 & \dots & \widetilde{\beta}_3 \\
  \end{array}\right]
= \frac{\widetilde{\beta}_0+\widetilde{\beta}_3}{2} \mathbf{C}_1
+ \frac{\widetilde{\beta}_0-\widetilde{\beta}_3}{2} \mathbf{C}_\delta
\]
where
\[
\mathbf{C}_1 =\left[ \begin{array}{cccccccc}  1 & 0 &0 & 0 & 0 & 0 & 0 & 0 \\
                                          0 & 0 &0 & 0 & 0 & 0 & 0 & 0 \\
                                          0 & 0 &0 & 0 & 0 & 0 & 0 & 0 \\
                                          0 & 0 &0 & 1 & 0 & 0 & 0 & 0 \\
                                          0 & 0 &0 & 0 & 1 & 0 & 0 & 0 \\
                                          0 & 0 &0 & 0 & 0 & 0 & 0 & 0 \\
                                          0 & 0 &0 & 0 & 0 & 0 & 0 & 0 \\
                                          0 & 0 &0 & 0 & 0 & 0 & 0 & 1 \\
  \end{array}\right] \qquad
\mathbf{C}_\delta =\left[ \begin{array}{cccccccc}  1 & 0 &0 & 0 & 0 & 0 & 0 & 0 \\
                                          0 & 0 &0 & 0 & 0 & 0 & 0 & 0 \\
                                          0 & 0 &0 & 0 & 0 & 0 & 0 & 0 \\
                                          0 & 0 &0 & -1 & 0 & 0 & 0 & 0 \\
                                          0 & 0 &0 & 0 & 1 & 0 & 0 & 0 \\
                                          0 & 0 &0 & 0 & 0 & 0 & 0 & 0 \\
                                          0 & 0 &0 & 0 & 0 & 0 & 0 & 0 \\
                                          0 & 0 &0 & 0 & 0 & 0 & 0 & -1 \\
  \end{array}\right]
\]
Since $\mathbf{Z}$ is stable, for all symmetric matrix $\mathbf{Y}$, that the solution $\mathbf{X}$
of $\mathbf{Z}^T\mathbf{X} + \mathbf{X}\mathbf{Z}+\mathbf{Y}=0$  is
\[
\mathbf{X} = \int_0^\infty \mathrm{e}^{s\mathbf{Z}^T}\mathbf{Y}\mathrm{e}^{s\mathbf{Z}}\mathrm{d}{s}
\]
and is clearly linear in $\mathbf{Y}$. Therefore, calling $\mathbf{S}_1$ and $\mathbf{S}_\delta$ solutions corresponding to
$\mathbf{C}_1$ and $\mathbf{C}_\delta$, we have
\[
\mathbf{S} = \frac{\widetilde{\beta}_0+\widetilde{\beta}_3}{2} \mathbf{S}_1
+ \frac{\widetilde{\beta}_0-\widetilde{\beta}_3}{2} \mathbf{S}_\delta
\]

As an additional remark, note that if $\mathbf{Z}^T+\mathbf{Z}+\mathbf{C}_1 =0$, as in the case of the matrix $\mathbf{Z}$
in Section \ref{sect:two-noises-model} with $\kappa=0$, then
\[
\int_0^\infty \mathrm{e}^{s\mathbf{Z}^T}\mathbf{C}_1 \mathrm{e}^{s\mathbf{Z}}\mathrm{d}{s}
-\int_0^\infty \mathrm{e}^{s\mathbf{Z}^T} (\mathbf{Z}^T+\mathbf{Z})\mathrm{e}^{s\mathbf{Z}}\mathrm{d}{s}
= \unit_{8}
\]
and so
\[
\mathbf{S} = \frac{\widetilde{\beta}_0+\widetilde{\beta}_3}{2} \unit_8
+ \frac{\widetilde{\beta}_0-\widetilde{\beta}_3}{2} \mathbf{S}_\delta
\]

\section*{Appendix B. Inequalities}

This section contains the proofs of the elementary inequalities that we need to establish the sign of eigenvalues
of matrices $\widetilde{\mathbf{S}}_{\text{\tiny red}}$.

\begin{lemma}\label{lem:1noise-lambdaaq-pos}
For all $|\kappa|<1$, $\widetilde{\beta} > 1$ and $0<g^2<(1-\kappa^2)/|\kappa|$ we have
\begin{eqnarray}
& & 6 - (2 + g^2)^2 \kappa^2 + (-2 + g^4) \kappa^4) >0 \label{eq:1-noise-lambda2-coeff-beta} \\
& & 4+\left(2-g^2\right)^2 \kappa^2-\left(3 g^4+14\right) \kappa^4 + 2 \left(g^4+4\right) \kappa^6-2 \kappa^8 > 0
\label{eq:1-noise-lambda2-coeff} \\
& & 2  \widetilde{\beta}^2 \left(2-g^2 \kappa^2\right)+2 \left(1-\kappa^2\right) \left(g^2 \kappa^4+\left(g^4+g^2+2\right)
   \kappa^2-2\right) >0 \label{eq:1-noise-lambda-coeff}
\end{eqnarray}
\end{lemma}

\begin{proof} First of all note that, since $|\kappa|<1$, the inequality $0<g^2<(1-\kappa^2)/|\kappa|$ implies
the weaker one $0<g^2<(1-\kappa^2)/\kappa^2$ that will be used in the proof. \\
To prove the inequality \eqref{eq:1-noise-lambda2-coeff-beta}, first write the left-hand side as
\(
6-4\kappa^2-2\kappa^4 -4g^2\kappa^2 - g^4\kappa^2\left( 1-\kappa^2\right).
\)
Therefore, by $g^2\kappa^2<(1-\kappa^2)$ and $g^4\kappa^2<(1-\kappa^2)^2$, we have
\begin{eqnarray*}
6 - (2 + g^2)^2 \kappa^2 + (-2 + g^4) \kappa^4)
& \geq & 2(1-\kappa^2)(3+\kappa^2) - 4(1-\kappa^2)-\left( 1-\kappa^2\right)^3 \\
& = & (1-\kappa^2)\left(1+4\kappa^2-\kappa^4 \right) > 0.
\end{eqnarray*}

Collecting powers of $g$, we can write \eqref{eq:1-noise-lambda2-coeff} as
\begin{equation}\label{eq:g-series}
4+4 \kappa^2-14 \kappa^4+8 \kappa^6 -2 \kappa^8-4 \kappa^2 g^2 +\left(1-2 \kappa^2\right)\left(1-\kappa^2\right)\kappa^2g^4
\end{equation}
Therefore, if $\kappa^2\geq 1/2$, by the inequality $0<g^2\kappa^2<1-\kappa^2$,
\eqref{eq:g-series} is bigger than
\begin{eqnarray*}
& & 4+4 \kappa^2-14 \kappa^4+8 \kappa^6 -2 \kappa^8 -4(1-\kappa^2) + \left(1-2 \kappa^2\right)\left(1-\kappa^2\right)^3 \\
& = & (1 - \kappa^2)  (1 + 4 \kappa^2 - \kappa^4)   > 0
\end{eqnarray*}
Otherwise, if $\kappa^2\leq 1/2$, note that the last term of \eqref{eq:g-series} is positive and hence it is bigger than
\begin{eqnarray*}
&&4+4 \kappa^2-14 \kappa^4+8 \kappa^6 -2 \kappa^8-4(1-\kappa^2)\\
&=&2\kappa^2(1-\kappa^2)(4-3\kappa^2+\kappa^4)>0
\end{eqnarray*}
and the inequality \eqref{eq:1-noise-lambda2-coeff} is proved.

To check \eqref{eq:1-noise-lambda-coeff}, note that, since $2-g^2 \kappa^2>2-(1-\kappa^2)>1$ and $\kappa^4g^2<g^2\kappa^2<1$
the left-hand side is bigger than
\begin{eqnarray*}
&& 2\left(2-g^2 \kappa^2\right)+2 \left(1-\kappa^2\right) \left(g^2 \kappa^4+\left(g^4+g^2+2\right)\kappa^2-2\right) \\
& = & 2\kappa^2 \left(4 - 2\kappa^2 - \kappa^4 g^2 + (1 - \kappa^2) g^4\right)  \\
& > & 2\kappa^2 \left(4 - 2\kappa^2 -1\right)  >0
\end{eqnarray*}
This proves the inequality \eqref{eq:1-noise-lambda-coeff}.
\end{proof}

\section*{Acknowledgement}
AD FF and DP are members of GNAMPA-INdAM.
AD and FF  acknowledge the support of the MUR grant ``Dipartimento di Eccellenza 2023--2027'' of Dipartimento di Matematica, Politecnico di Milano and ``Centro Nazionale di ricerca in HPC, Big Data and Quantum Computing''.
DP has been supported by the MUR grant ``Dipartimento di Eccellenza 2023–2027'' of Dipartimento di Matematica, Universit\`a di Genova.
HJY has been supported by the Korean government grant MSIT (No. RS-2023-00244129).

The authors would like to thank Alain Joye who brought the problem to their attention after a seminar on the results of
\cite{HaackJoye} as well as for fruitful discussions.

\bigskip

{\footnotesize
Authors' addresses:
\begin{itemize}
\item[$^{(1)}$] Mathematics Department, Politecnico di Milano,
Piazza Leonardo da Vinci 32, I - 20133 Milano, Italy 
\item[$^{(2)}$] Mathematics Department, University of Genova,
Via Dodecaneso 35, I - 16146 Genova, Italy
\item[$^{(3)}$] Department of Applied Mathematics and Institute for Integrated Mathematical Sciences, Hankyong National University, 327 Jungang-ro, Anseong-si, Gyeonggi-do 17579, Korea
\end{itemize}
}



\begin{thebibliography}{100}


\bibitem{AcFaQu}
L. Accardi, F. Fagnola, and R. Quezada,
On three new principles in non-equilibrium statistical mechanics and Markov semigroups of weak coupling limit type.
{\it Infin. Dimens. Anal. Quantum Probab. Relat. Top.},  {\bf 19} (2), 1650009 (2016).
https://doi.org/10.1142/S0219025716500090

\bibitem{AlLe}
R. Alicki, and K. Lendi, {\it Quantum Dynamical
Semigroups and Applications}, Springer-Verlag Berlin Heidelberg 2007.

\bibitem{AFPOSID22}
J. Agredo, F. Fagnola, and D. Poletti,
Gaussian Quantum Markov Semigroups on a One-Mode Fock Space:
Irreducibility and Normal Invariant States.
{\it Open Sys. Inf. Dyn.} {\bf 28}, No. 1 (2021) 2150001
https://doi.org/10.1142/S1230161221500013

\bibitem{AFPMJM}
J. Agredo, F. Fagnola and D. Poletti,
The Decoherence--Free Subalgebra of Gaussian Quantum Markov Semigroups,
{\it Milan J. Math.} {\bf 90} (2022) 257–289 \\
https://doi.org/10.1007/s00032-022-00355-0  \quad
{\tt arXiv.2112.13781}

\bibitem{BCFS}
F. Benatti, F. Carollo, R. Floreanini and J. Surace,
Long-Lived Mesoscopic Entanglement Between Two
Damped Infinite Harmonic Chains
{\it J. Stat. Phys.} {\bf 168} (2017) 620--651
https://doi.org/10.1007/s10955-017-1817-8

\bibitem{BFP}
F. Benatti, R. Floreanini, M. Piani,
Environment Induced Entanglement in Markovian Dissipative Dynamics,
{\it J. Phys. A} {\bf 91} 7 070402 (2003).
https://10.1103/PhysRevLett.91.070402

\bibitem{Raja-Tiju}
B. V. Rajarama Bhat, Tiju Cherian John, R. Srinivasan,
Infinite mode quantum Gaussian states,
{\it Rev. Math. Phys.} {\bf 31} (09) 1950030 (2019).
https://doi.org/10.1142/S0129055X19500302

\bibitem{BeFl}
F. Benatti, R. Floreanini, Entangling oscillators through environment noise,
{\it J. Phys. A} {\bf 39} (2006) 2689–2699.
https://doi:10.1088/0305-4470/39/11/009

\bibitem{BoRCUr}
J.R. Bola\~{n}os-Serv{\'\i}n, J.I. Rios-Cangas, A. Uribe,
The Fast Recurrent Subspace on an N-Level Quantum Energy Transport Model
{\it Open Syst. Inf. Dyn.}  {\bf 31}, No 1, 2450002 (2024).
https://doi.org/10.1142/S1230161224500021

\bibitem{DeVaVe}
B. Demoen, P. Vanheuverzwijn and A. Verbeure, Completely positive maps on the CCR-algebra.
\textit{Lett. Math. Phys.} \textbf{2}, 161 -- 166 (1977).
https://doi.org/10.1007/BF00398582

\bibitem{CrDVMoRo}
V. Crismale, S. Del Vecchio, T. Monni, S. Rossi,
Freedman's Theorem for Unitarily Invariant States on the CCR Algebra,
{\it Commun. Math. Phys.}  {\bf 405}, 40 (2024).
https://doi.org/10.1007/s00220-024-04932-9

\bibitem{FP-IDAQP22}
F. Fagnola and D. Poletti,
On Irreducibility of Gaussian Quantum Markov Semigroups.
{\it Infin. Dimens. Anal. Quantum Probab. Relat. Top.} {\bf 25}, No. 04, 2240001 (2022)
https://doi.org/10.1142/S021902572240001X

\bibitem{FP-IDAQP24}
F. Fagnola and D. Poletti,
A note on invariant states of Gaussian quantum Markov semigroups
{\it Infin. Dimens. Anal. Quantum Probab. Relat. Top.} {\bf 27},  (2024).
To appear

\bibitem{GJN}
J.E. Gough, M.R. James, H.I. Nurdin, Squeezing components in linear quantum feedback networks.
{\it Phys. Rev. A} {\bf 81}, 023804 (2010).
https://link.aps.org/doi/10.1103/PhysRevA.81.02380

\bibitem{HaackJoye}
Haack, G., Joye, A. Perturbation Analysis of Quantum Reset Models.
{\it J. Stat. Phys.} {\bf 183}, 17 (2021).
https://doi.org/10.1007/s10955-021-02752-y

\bibitem{JPW}
V. Jak\v{s}i\'c, C.-A. Pillet and M. Westrich,
Entropic Fluctuations of Quantum Dynamical Semigroups,
{\it J. Stat. Phys.} {\bf 154} 153--187, (2014).
https://doi.org/10.1007/s10955-013-0826-5

\bibitem{KRP-cosa}
K. R.~Parthasarathy, What is a Gaussian State?, {\sl Commun. Stoch. Anal.} \textbf{4}, (2010) 143--160.
https://repository.lsu.edu/cosa/vol4/iss2/2/

\bibitem{LJB}
B. Longstaff, M.G. Jabbour, J.B. Brask,
Bosonic autonomous entanglement engines with weak bath coupling are impossible
{\it Phys. Rev. A} {\bf 108}, 032209 (2023).
https://doi.org/10.1103/PhysRevA.108.032209

\bibitem{Math24}
Math24.net, Routh-Hurwitz Criterion,
https://math24.net/routh-hurwitz-criterion.html

\bibitem{Par13}
K.R. Parthasarathy,  The Symmetry Group of Gaussian States in $L(\mathbb{R}^n)$. In: Shiryaev, A., Varadhan, S., Presman, E. (eds) \textit{Prokhorov and Contemporary Probability Theory}. Springer Proceedings in Mathematics \& Statistics, vol 33. pp 349--369. Springer, Berlin, Heidelberg.
https://doi.org/10.1007/978-3-642-33549-5\_21

\bibitem{Po2022}
D. Poletti, Characterization of Gaussian Quantum Markov Semigroups,
{\it Infin. Dimens. Anal. Quantum Probab. Relat. Top.} {\bf 25}  (3), 2250014 (2022).
https://doi.org/10.1142/S021902572250014X

\bibitem{RoRouSel}
R. Robin, P. Rouchon, L.-A. Sellem,
Convergence of bipartite open quantum systems stabilized by reservoir engineering,
https://doi.org/10.48550/arXiv.2311.10037

\bibitem{SimonR}
R. Simon, Peres-Horodecki Separability Criterion for Continuous Variable Systems,
{\it Phys. Rev. Lett.} {\bf 84}, {2726--2729} (2000).
https://link.aps.org/doi/10.1103/PhysRevLett.84.2726

\bibitem{Teret}
A. E. Teretenkov,  Irreversible quantum evolution with quadratic generator: Review.
{\it Infinn. Dimens. Anal. Quantum Probab. Relat. Top.} {\bf 22}, 1930001 (2019).
https://doi.org/10.1142/S0219025719300019

\bibitem{TicVio}
F. Ticozzi and L. Viola, Steady-state entanglement by engineered quasi-local
markovian dissipation. {\it Quantum Information and Computation} {\bf 14}, 0265--0294 (2014).
https://doi.org/10.26421/QIC14.3-4-5

\bibitem{ViWe}
G, Vidal, R. F. Werner, Computable measure of entanglement,
{\it Phys. Rev. Lett.} {\bf 65} 032314 (2002).
https://link.aps.org/doi/10.1103/PhysRevA.65.032314

\bibitem{WeWo}
R.F. Werner, M.M. Wolf, Bound Entangled Gaussian States.
{\it Phys. Rev. Lett.}  {\bf 86} (16) 3658 (2001)
https://doi.org/10.1103/PhysRevLett.86.3658

\bibitem{WoEiPl}
M.M.Wolf, J. Eisert and Martin B. Plenio, Entangling Power of Passive Optical Elements,
{\it Phys. Rev. Lett.} {\bf 90} (4)  047904 (2003)
https://doi.org/10.1103/PhysRevLett.90.047904

\end{thebibliography}
\end{document}